\documentclass[preprint]{elsarticle}

\usepackage{lineno,hyperref}
\usepackage{multirow,xcolor}
\usepackage[vlined,algo2e,ruled]{algorithm2e}
\usepackage{caption}
\usepackage{booktabs}
\modulolinenumbers[5]

\newtheorem{theorem}{Theorem}
\newtheorem{lemma}{Lemma}
\newproof{proof}{Proof}

\usepackage{amssymb}

\journal{Computational Geometry}

\newcommand{\no}[1]{}

\newcommand{\Oh}[1]{\ensuremath{\mathcal{O}\!\left({#1}\right)}}
\renewcommand{\lg}{\log}
\newcommand{\rank}{\ensuremath{\mathsf{rank}}}
\newcommand{\select}{\ensuremath{\mathsf{select}}}
\newcommand{\match}{\ensuremath{\mathsf{match}}}
\newcommand{\parent}{\ensuremath{\mathsf{parent}}}
\newcommand{\first}{\ensuremath{\mathsf{first}}}
\newcommand{\last}{\ensuremath{\mathsf{last}}}
\newcommand{\next}{\ensuremath{\mathsf{next}}}
\newcommand{\prev}{\ensuremath{\mathsf{prev}}}
\newcommand{\mate}{\ensuremath{\mathsf{mate}}}
\newcommand{\vertex}{\ensuremath{\mathsf{vertex}}}
\newcommand{\degree}{\ensuremath{\mathsf{degree}}}
\newcommand{\neighbour}{\ensuremath{\mathsf{neighbour}}}
\newcommand{\listing}{\ensuremath{\mathsf{listing}}}
\newcommand{\face}{\ensuremath{\mathsf{face}}}

\newcommand{\asgn}{\mathrel{=}}

\begin{document}

\begin{frontmatter}

\title{Fast and Compact Planar Embeddings\tnoteref{wadsnote}}
\tnotetext[wadsnote]{A previous version of this paper appeared in the 15th
  Algorithms and Data Structures Symposium (WADS 2017)\cite{Ferres2017}.}


\author[udd]{Leo Ferres}
\ead{lferres@udd.cl}

\author[dcc,cebib]{Jos\'e Fuentes-Sep\'ulveda}
\ead{jfuentess@dcc.uchile.cl}

\author[udp,cebib]{Travis Gagie\corref{cor1}}
\ead{travis.gagie@mail.udp.cl}

\author[dal]{Meng He}
\ead{mhe@cs.dal.ca}

\author[dcc,cebib]{Gonzalo Navarro}
\ead{gnavarro@dcc.uchile.cl}

\address[udd]{Faculty of Engineering, Universidad del Desarrollo \& Telef\'onica
  I+D, Santiago, Chile}
\address[dcc]{Department of Computer Science, University of Chile,
  Santiago, Chile.}
\address[cebib]{Center of Biotechnology and Bioengineering, University of Chile,
  Santiago, Chile.}
\address[udp]{School of Computer Science and Telecommunications, Diego Portales
  University, Santiago, Chile}
\address[dal]{Faculty of Computer Science, Dalhousie University, Halifax, Canada}

\cortext[cor1]{Corresponding author}

\begin{abstract}
There are many representations of planar graphs, but few are as elegant as
Tur\'an's (1984): it is simple and practical, uses only 4 bits per edge, can
handle self-loops and multi-edges, and can store any specified embedding. 
Its main disadvantage
has been that ``it does not allow efficient searching'' (Jacobson, 1989).  In
this paper we show how to add a sublinear number of bits to Tur\'an's
representation such that it supports fast navigation while retaining simplicity.
As a consequence of the inherited simplicity, we offer the first efficient
parallel construction of a compact encoding of
a planar graph embedding. Our experimental results show that the resulting
representation uses about 6 bits per edge in practice, supports basic 
navigation operations within a few microseconds, and can be built sequentially 
at a rate below 1 microsecond per edge, featuring a linear speedup with a
parallel efficiency around 50\% for large datasets.
\end{abstract}


\end{frontmatter}


\section{Introduction}
\label{sec:introduction}

The rate at which we store data is increasing even faster than the speed and
capacity of computing hardware. Thus, if we want to use efficiently what we 
store, we need to represent it in better ways. The surge in the number
and complexity of the maps we want to have available on mobile devices is
particularly pronounced and has resulted in a bewildering number of ways to
store planar graphs. Each of these representations has its disadvantages,
however: some do not support fast navigation, some are large, some cannot
represent multi-edges or certain embeddings, and some are costly to build.
In this paper we introduce a compact representation of planar graph 
embeddings that addresses all these issues, and demonstrate its practicality.

More concretely, as described in Section~\ref{sec:related}, a planar embedding 
with $n$ nodes and $m$ edges can be represented in \(m \lg 12
\approx 3.58 m\) bits 
\cite{Tutte1963}, which has been matched with $o(m)$-bit redundancy with a 
structure that in addition supports efficient
navigation \cite{BlellochFarzan2010}. The structure is, however, complex and no
implementation has been attempted. The much simpler representation of Tur\'an 
\cite{Turan1984} uses $4m$ bits, which is still close to the lower bound, but 
it does not support navigation. The other existing representations require 
more than $4m$ bits for general planar embeddings, some restrict the embeddings
where they apply, and most have complicated construction algorithms.
The majority of these constructions cannot be parallelized, and the others
require $\Oh{m\log m}$ work. 

Our contribution in this paper is threefold: 
\begin{enumerate}
\item We show how to add \(o (m)\)
bits to Tur\'an's representation such that it supports fast navigation. We can
list the edges incident to any vertex in clockwise or counter-clockwise order 
using constant time per edge, including starting the enumeration at any desired
neighbour. As a consequence, we can also list the nodes on a face in constant 
time per node.
We can also find a vertex's degree in time $\Oh{f(m)}$ for any $f(m) \in
\omega(1)$, and determine whether two vertices are neighbours in $\Oh{f(m)}$ 
time for any $f(m) \in \omega(\log m)$.
\item We give a parallel algorithm that builds our data structure from any 
spanning tree of the planar embedding, in $\Oh{m}$ work and $\Oh{\log m}$ span 
($\Oh{\log^2 m}$ span to support the neighbour query). This is the 
first linear-work practical parallel algorithm for building compact 
representations of planar graphs. 
\item We implement and experimentally evaluate the space, query, and 
construction performance of our representation. In practice, our structure uses
less than $6m$ bits, performs navigation operations within a few microseconds,
and can be built sequentially at a rate below 1 microsecond per edge. The
parallel algorithm scales linearly, with an efficiency around 50\% for large
datasets, with up to 24 processors. 
\end{enumerate}

Summarizing, we offer the first simple compact representation of planar
embeddings, which is easy to program, uses little space, and is efficiently 
built and navigated. Our structure is thousands of times faster than the 
classical one when compression makes our representation fit in main memory.
We leave the code publicly available at 
\url{http://www.dcc.uchile.cl/~jfuentess/pemb/}.

Tur\'an chooses an arbitrary spanning tree of the graph, roots it at a vertex on
the outer face and traverses it, writing its balanced-parentheses representation
as he goes and interleaving that sequence with another over a different binary
alphabet, consisting of an occurrence of one character for the first time he
sees each edge not in the tree and an occurrence of the other character for the
second time he sees that edge.  These two sequences can be written as three
sequences over \(\{0, 1\}\): one of length \(2 n - 2\) encoding the
balanced-parentheses representation of the tree; one of length \(2 m - 2n + 2\)
encoding the interleaved sequence; and one of length \(2 m\) indicating how they
are interleaved.  Our extension of this representation is based on the
observation that the interleaved sequence encodes the balanced-parentheses
representation of the complementary spanning tree of the dual of the graph.  By
adding a sublinear number of bits to each balanced-parentheses representation,
we can support fast navigation in the trees, and by storing the sequence
indicating the interleaving as a bitvector with support for operations $\rank$
and $\select$ \cite{Jacobson1989}, we can support fast navigation in the graph.

Section~\ref{sec:related} surveys the related work on compact representations
of planar embeddings. Section~\ref{sec:preliminaries} describes 
bitvectors and the balanced-parentheses representation of trees, which are the 
building blocks of our extension of Tur\'an's representation. We also describe 
the model of parallelism we use in our construction algorithms. In 
Section~\ref{sec:trees} we prove the observation mentioned above on Tur\'an's
interleaved sequence. In
Section~\ref{sec:structure} we describe our data structure and how we implement
queries.  Section~\ref{sec:parallel} describes our parallel construction
algorithm and discusses some implementation issues. In
Section~\ref{sec:experiments} we describe our experiments on space, query and
construction performance, and discuss the results. Finally, in 
Section~\ref{sec:future} we present our conclusions and future work directions. 

\section{Related work} \label{sec:related}

Tutte~\cite{Tutte1963} showed that representing a specified embedding of a
connected planar multi-graph with $n$ vertices and $m$ edges requires \(m \lg 12
\approx 3.58 m\) bits in the worst case.  Tur\'an~\cite{Turan1984} gave a very
simple representation that uses \(4 m\) bits. Jacobson~\cite{Jacobson1989}
argued that this representation ``does not allow fast searching'' and
(introducing techniques that we will apply to Tur\'an's representation)
proposed 
one that instead uses $\Oh{m}$ bits and supports fast navigation, based on
book embeddings \cite{1989y}. Munro and Raman \cite{MR01} estimated that 
Jacobson's representation uses $64n$ bits and proposed one using $2m+8n+o(m)$ 
bits that retains fast navigation, still based on the same book embeddings
(but this does not handle self-loops). 
Keeler and Westbrook~\cite{KeelerWestbrook1995} also noted that ``the constant 
factor in [Jacobson's] space bound is relatively large'' and gave a 
representation that uses \(m \lg 12 + \Oh{1}\) bits for planar graphs (not 
embeddings), as well as for planar embeddings containing either no self-loops 
or no vertices with degree 1; however, they again gave up fast navigation. 
Chiang, Lin and Lu~\cite{ChiangLinLu2005}, improving previous work by Chuang
{\it et al.}~\cite{CGHKL98}, gave a representation (without allowing self-loops) that
uses \(2 m + 3 n + o (m)\) bits with fast navigation, based on so-called orderly
spanning trees. However, although all planar graphs can be represented with 
orderly spanning trees, some planar embeddings cannot. For simple planar
embeddings (i.e., no self-loops nor multiple edges, thus $m \le 3n$), their 
space decreases to $2n+2m+o(m) \le 4m+o(m)$ on connected graphs. Barbay {\it et
  al.}~\cite{bchm2012} gave a data structure that uses $\Oh{m}$ bits to
represent  
simple planar graphs with fast navigation, based on realizers of 
planar triangulations~\cite{s1990}. Still, their constant is relatively large, 
$18n+o(m)$. Finally, Blelloch and Farzan~\cite{BlellochFarzan2010}, extending
the work of Blandford {\it et al.}.~\cite{bbk2003}, matched for the first time
Tutte's lower bound on
general planar embeddings, with a structure that uses \(m \lg 12 + o (m)\) bits
and supports fast navigation. Their structure is based on small vertex
separators~\cite{lt1979}. They can also represent any planar graph within its
lower-bound space plus a sublinear redundancy, even when the exact lower bound
is unknown for general planar graphs \cite{BGHPS06}. While Blelloch and Farzan
closed the problem in theoretical terms, their representation is complicated 
and has not been implemented. 
Other authors \cite{HKL00,aleardi2005succinct,aleardi2008succinct,fusy2008dissections,YN10}
have considered special kinds of planar graphs, notably tri-connected planar
graphs and triangulations.  
We refer the reader to Munro and Nicholson's \cite{MunroNicholson2016} and 
Navarro's \cite[Chapter~9]{Navarro2016} recent surveys for further discussion 
of compact data structures for graphs.

Most of the navigable representations we have mentioned require 
complicated construction algorithms, generally defying a parallel 
implementation.
It is not known how to compute a book embedding \cite{1989y} in parallel, 
which is necessary to build the representations of Jacobson and of Munro and 
Raman. There are also no parallel algorithms to build orderly spanning trees
\cite{ChiangLinLu2005}, necessary for the representation of Chiang {\it et al.}
Its predecessor \cite{CGHKL98} uses instead a triangulation and a canonical
ordering; for the latter there is only a CREW construction running in 
$\Oh{\log^4 n}$ time with $n^2$ processors \cite{He1993}. 
As for the vertex separators \cite{lt1979} required to build the representation
of Blandford {\it et al.}\ and of Blelloch and Farzan, Kao {\it et al.}~\cite{ktt1995} 
designed a linear-work, logarithmic-span algorithm for computing a cycle 
separator of a planar graph. However, the construction of these representations
of planar embeddings decompose the input graph by repeatedly computing
separators until each piece is sufficiently small. This increases the total 
work to $\Oh{n \log n}$ even if this optimal parallel algorithm is used.
The best linear-work parallel algorithms \cite{kfhr1994} for building the
realizers \cite{s1990} used in the construction of Barbay {\it et al.}'s 
representation have $\Oh{\log n}$ span in the expected case but
$\Oh{\log n \log \log n}$ deterministic span.

\section{Preliminaries}
\label{sec:preliminaries}

\subsection{Bitvectors and parentheses}

A bitvector is a binary string that supports the queries $\rank$ and $\select$
in addition to random access, where \(\rank_b (i)\) returns the number of bits
set to $b$ in the prefix of length $i$ of the string and \(\select_b (j)\)
returns the position of the $j$th bit set to $b$.  For convenience, we define
\(\select_b (0) = 0\). 

It is possible to represent a bitvector of length $\ell$ 
in \(\ell + o(\ell)\) bits and support random access, $\rank$ and $\select$ in 
constant time \cite{Jacobson1989,Cla96,Mun96}. Furthermore, if the bitvector
has $k$ 1s, it can be represented in $\log {\ell \choose k} + o(\ell)$ bits
\cite{RRR07}, which is $\ell H(k/\ell)+o(\ell) = k\log(\ell/k)+\Oh{k}+o(\ell)$,
with $H(x)=-x\log x-(1-x)\log(1-x)$. All our logarithms are to the base 2 
unless otherwise stated.

By adding some further operations on the bitvectors, we can represent an 
ordered tree or forest of $t$ vertices using \(2 t + o (t)\) bits and support 
natural navigation queries in constant time. One of the most popular such 
representations is a string of balanced parentheses: we traverse each tree 
from left to right, writing an opening parenthesis when we first visit a vertex
(starting at the root) and a closing parenthesis when we leave it for the last 
time (or, in the case of the root, when we finish the traversal). We can encode
the string of parentheses as a bitvector of length $2t$, with 0s encoding 
opening parentheses and 1s encoding closing parentheses. By adding $o(t)$ 
further bits, we can support in constant time, among others, the following 
queries used by our solution \cite{MR01,GRRR06,NS14}:
\begin{itemize}
\item \(\match (i)\) locates the position of the parenthesis matching the $i$th
parenthesis in the bitvector (i.e., finds the other parenthesis referring to the same node);
\item \(\parent (v)\) returns the parent of $v$, or 0 if $v$ is the root of its
tree. Nodes $v$ and $\parent(v)$ are represented as their pre-order rank in
the traversal.
\end{itemize}

\subsection{Parallel computation model}

As we focus on practical parallel algorithms, we describe and analyze our
construction using the {\em Dynamic Multithreading (DyM) Model} 
\cite{Cormen2009} (we nevertheless express our final results in terms of the
PRAM model as well). In the DyM model, a multithreaded computation is modelled 
as a directed acyclic graph (DAG) where vertices are instructions and edge 
$(u,v)$ represents precedence between instructions $u$ and $v$. The model
is based on two parameters of the multithreaded computation: its {\em work}
$T_1$ and its {\em span} $T_\infty$. The work is the running time on a single
thread, that is, the number of nodes (i.e., instructions) in the DAG, assuming
each instruction takes constant time. The span is the length of the longest path
in the DAG, that is,
the intrinsically sequential part of the computation. The time $T_p$
needed to execute the computation on $p$ threads then has complexity
$\Theta(T_1/p+T_\infty)$, which can be reached with a greedy scheduler. The
improvement of a multithreaded computation using $p$ threads is called {\em
  speedup}, $T_1/T_p$. The upper bound on the achievable speedup,
$T_1/T_\infty$, is called {\em parallelism}.  Finally, the {\em efficiency} is
defined as $T_1 /pT_p$ and can be interpreted as the percentage of improvement
achieved by using $p$ cores or how close we are to the linear speedup. In the
DyM model, the workload of the threads is balanced by using the {\em
  work-stealing} algorithm \cite{Blumofe:1999:SMC:324133.324234}.

To describe parallel algorithms in the DyM model, we augment sequential
pseudocode with three keywords. The {\bf spawn} keyword, followed by a procedure
call, indicates that the procedure should run in its own thread and may thus be
executed in parallel to the thread that spawned it. The {\bf sync} keyword
indicates that the current thread must wait for the termination of all threads
it has spawned. Finally, {\bf parfor} is ``syntactic sugar'' for {\bf spawn}ing
one thread per iteration in a for loop, thereby allowing these iterations to run
in parallel, followed by a {\bf sync} operation that waits for all iterations to
complete. In practice, the {\bf parfor} keyword is implemented by halving the
range of loop iterations, spawning one half and using the current procedure to
process the other half recursively until reaching one iteration per range. After
that, the iterations are executed in parallel. Therefore, this implementation
adds an overhead bounded above by the logarithm of the number of loop
iterations. We include such overheads in our complexities.

\section{Spanning trees of planar graphs}
\label{sec:trees}

It is well known \cite{Biggs1971,Eppstein2003,RileyThurston2006} 
that for any spanning tree $T$ of a connected planar graph $G$,
the edges dual to $T$ are a spanning tree $T^*$ of the dual of $G$, with $T$ and
$T^*$ interdigitating; see Figure~\ref{fig:trees} for an illustration (including
multi-edges and a self-loop). If we
choose $T$ as the spanning tree of $G$ for Tur\'an's representation, then we
store a 0 and a 1, in that order, for each edge in $T^*$.  We now show that
these bits encode a traversal of $T^*$.

\begin{figure}[t!]
\centering
\resizebox{.6\textheight}{!}
{\begin{tabular}{c@{\hspace{10ex}}c}
& \multirow{3}{*}
{ \sf
\begin{tabular}{r@{\hspace{3ex}}c@{\hspace{2ex}}c@{\hspace{2ex}}c}
& \textcolor{red}{$T$} & \(G - T\) & \textcolor{blue}{$T^*$} \\
\hline\\[-.5ex]
1 & & (1, 3) & \textcolor{blue}{(A, B)} \\
2 & \textcolor{red}{(1, 2)} & & \\
3 & \textcolor{red}{(2, 3)} & & \\
4 & & (1, 3) & \textcolor{blue}{(A, B)} \\
5 & \textcolor{red}{(2, 3)} & & \\
6 & \textcolor{red}{(2, 4)} & & \\
7 & & (4, 8) & \textcolor{blue}{(A, C)} \\
8 & \textcolor{red}{(2, 4)} & & \\
9 & & (2, 6) & \textcolor{blue}{(C, D)} \\
10 & \textcolor{red}{(1, 2)} & & \\
11 & \textcolor{red}{(1, 5)} & & \\
12 & \textcolor{red}{(5, 6)} & & \\
13 & & (2, 6) & \textcolor{blue}{(C, D)} \\
14 & & (6, 8) & \textcolor{blue}{(C, E)} \\
15 & \textcolor{red}{(5, 6)} & & \\
16 & & (5, 7) & \textcolor{blue}{(E, F)} \\
17 & \textcolor{red}{(1, 5)} & & \\
18 & \textcolor{red}{(1, 7)} & & \\
19 & & (5, 7) & \textcolor{blue}{(E, F)} \\
20 & \textcolor{red}{(7, 8)} & & \\
21 & & (6, 8) & \textcolor{blue}{(C, E)} \\
22 & & (4, 8) & \textcolor{blue}{(A, C)} \\
23 & & (7, 8) & \textcolor{blue}{(A, G)} \\
24 & \textcolor{red}{(7, 8)} & & \\
25 & & (7, 8) & \textcolor{blue}{(A, G)} \\
26 & \textcolor{red}{(1, 7)} & & \\
27 & & (1, 1) & \textcolor{blue}{(A, H)} \\
28 & & (1, 1) & \textcolor{blue}{(A, H)}
\end{tabular}} \\[-3ex]
\includegraphics[height=.3\textheight]{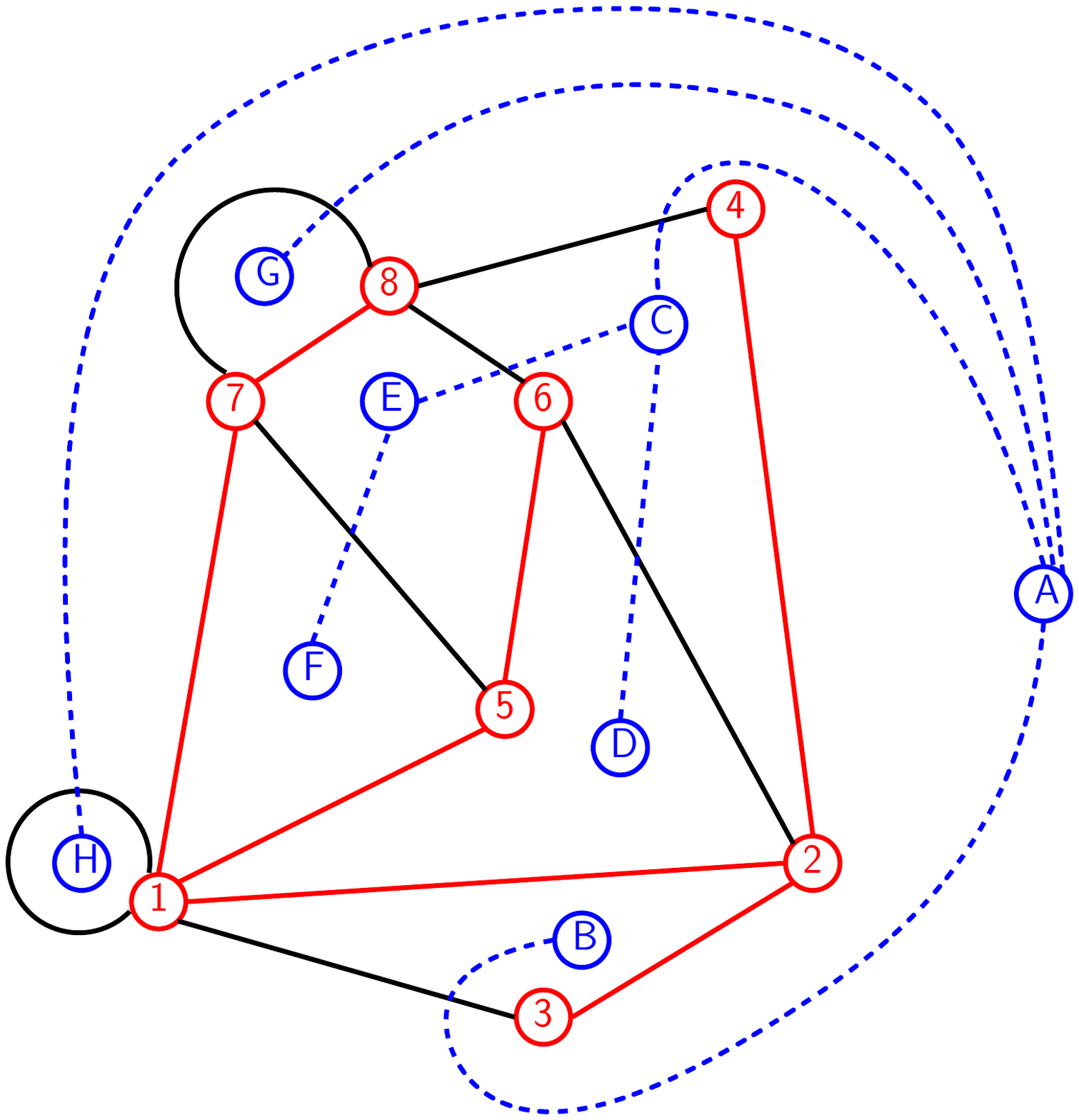} & \\[10ex]
\includegraphics[height=.3\textheight]{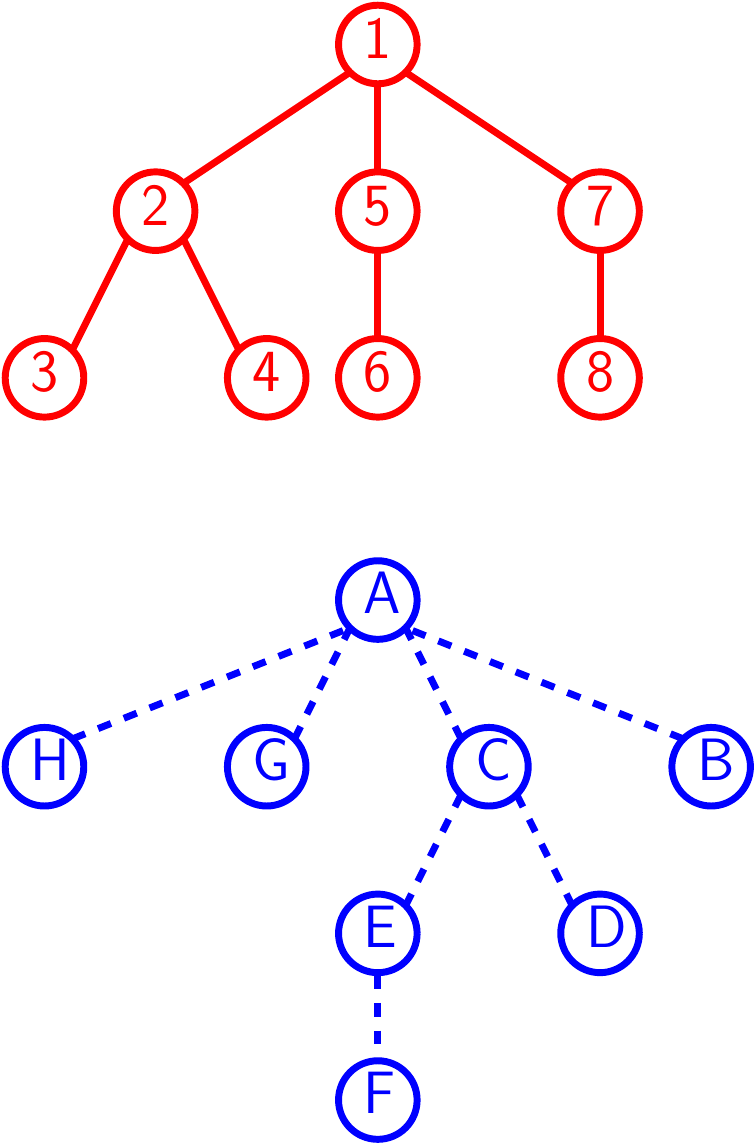} & \\[10ex]
\end{tabular}}
\vspace*{-1cm}
\caption{{\bf Top left:} A planar embedding of a planar graph $G$, with a
  spanning tree $T$ of $G$ shown in red and the complementary spanning tree
  $T^*$ of the dual of $G$ shown in blue with dashed lines.  {\bf Bottom left:}
  The two spanning trees, with $T$ rooted at the vertex {\sf 1} on the outer
  face and $T^*$ rooted at the vertex {\sf A} corresponding to the outer face.
  {\bf Right:} The list of edges we process while traversing $T$ starting at
  {\sf 1} and processing edges in counter-clockwise order, with the edges in $T$
  shown in red and the ones in \(G - T\) shown in black; the edges of $T^*$
  corresponding to the edges in \(G - T\) are shown in blue.}
\label{fig:trees}
\end{figure}

\begin{lemma}
\label{lem:trees}
Consider any planar embedding of a planar graph $G$, any spanning tree $T$ of
$G$ and the complementary spanning tree $T^*$ of the dual of $G$.  If we perform
a depth-first traversal of $T$ starting from any vertex on the outer face of $G$
and always process the edges incident to the vertex $v$ we are visiting in
counter-clockwise order (starting from the edge immediately after the one to
$v$'s parent or, if $v$ is the root of $T$, from immediately after any incidence
of the outer face), then each edge not in $T$ corresponds to the next edge we
cross in a depth-first traversal of $T^*$.
\end{lemma}

\begin{proof}
Suppose the traversal of $T^*$ starts at the vertex of the dual of $G$
corresponding to the outer face of $G$.  We now prove by induction that the
vertex we are visiting in $T^*$ always corresponds to the face of $G$ incident
to the vertex we are visiting in $T$ and to the previous and next edges in
counter-clockwise order.

Our claim is true before we process any edges, since we order the edges starting
from an incidence of the outer face to the root of $T$.  Assume it is still true
after we have processed \(i < m\) edges, and that at this time we are visiting
$v$ in $T$ and $v^*$ in $T^*$.  First suppose that the \((i + 1)\)th edge \((v,
w)\) we process is in $T$.  We note that \(w \neq v\), since otherwise \((v,
w)\) could not be in $T$.  We cross from $v$ to $w$ in $T$, which is also
incident to the face corresponding to $v^*$.  Now \((v, w)\) is the previous
edge --- considering their counter-clockwise order at $w$, starting from \((v,
w)\) --- and the next edge (which is \((v, w)\) again if $w$ has degree 1) is
also incident to $v^*$.  This is illustrated on the left side of
Figure~\ref{fig:induction}.  In fact, the next edge is the one after \((v, w)\)
in a clockwise traversal of the edges incident to the face corresponding to
$v^*$.

Now suppose \((v, w)\) is not in $T$ and let $w^*$ be the vertex in $T^*$
corresponding to the face on the opposite side of \((v, w)\), which is also
incident to $v$.  We note that \(w^* \neq v^*\), since otherwise \((v, w)\)
would have to be in $T$.  We cross from $v^*$ to $w^*$ in $T^*$.  Now \((v, w)\)
is the previous edge --- this time still considering their counter-clockwise
order at $v$ --- and the next edge (which may be \((v, w)\) again if it is a
self-loop) is also incident to $w^*$.  This is illustrated on the right side of
Figure~\ref{fig:induction}.  In fact, the next edge is the one that follows
\((v, w)\) in a clockwise traversal of the edges incident to the face
corresponding to $w^*$.

Since our claim remains true in both cases after we have processed \(i + 1\)
edges, by induction it is always true.  In other words, whenever we should
process next an edge $e$ in $G$ that is not in $T$, we are visiting in $T^*$ one
of the vertices corresponding to the faces incident to $e$ (i.e., one of the
endpoints of the edge in the dual of $G$ that corresponds to $e$).  Since we
process each edge in $G$ twice, once at each of its endpoints or twice at its
unique endpoint if it is a self-loop, it follows that the list of edges we
process that are not in $T$, corresponds to the list of edges we cross in a
traversal of $T^*$.
\qed
\end{proof}

\begin{figure}[t]
\centering
\includegraphics[width=0.7\textwidth]{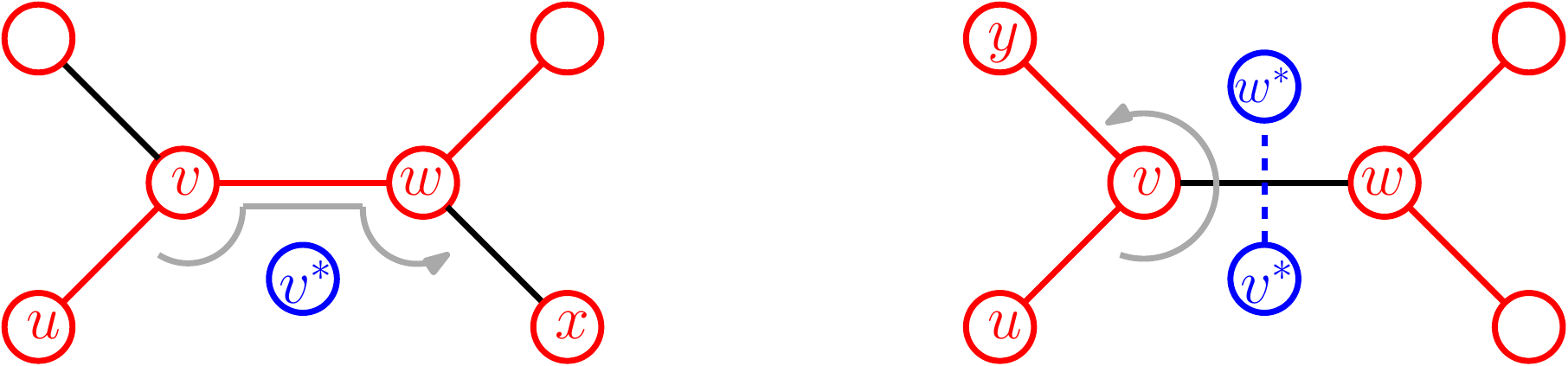}
\caption{{\bf Left:} If we process an edge \((v, w)\) in $T$, then we move to
  $w$ in our traversal of $T$ and the next edge, \((w, x)\) in this case, is
  also incident to the vertex $v^*$ we are visiting in our traversal of $T^*$.
  {\bf Right:} If \((v, w)\) is not in $T$, then in $T^*$ we move from $v^*$ to
  the vertex $w^*$ corresponding to the face on the opposite side of \((v, w)\)
  in $G$.  The next edge, \((v, y)\) in this case, is also incident to $w^*$.} 
\label{fig:induction}
\end{figure}

We process the edges in counter-clockwise order so that the traversals of $T$
and $T^*$ are from left to right and from right to left, respectively;
processing them in clockwise order would reverse those directions.  For example,
for the embedding in Figure~\ref{fig:trees}, if we start the traversal of the
red tree $T$ at vertex 1 and start processing the edges at \((1, 3)\), then we
process them in the order shown at the right of the figure.

\section{Data structure}
\label{sec:structure}

Our extension of Tur\'an's representation of a planar embedding of a connected
planar graph $G$ with $n$ vertices and $m$ edges consists of the following
components, which take \(4 m + o (m)\) bits: 
\begin{itemize}
\item a bitvector \(A [1..2 m]\) in which \(A [i]=1\) if and only if the $i$th
  edge we process in the traversal of $T$ described in Lemma~\ref{lem:trees}, is
  in $T$; 
\item a bitvector \(B [1..2 (n - 1)]\) in which \(B [i]=0\) if and only if the
  $i$th time we process an edge in $T$ during the traversal, is the first time
  we process that edge; 
\item a bitvector \(B^* [1..2 (m - n + 1)]\) in which \(B^* [i]=0\) if and only if
  the $i$th time we process an edge not in $T$ during the traversal, is
  the first time we process that edge. 
\end{itemize}
Notice $B$ encodes the balanced-parentheses representation of $T$, except that it
lacks the leading 0 and trailing 1 encoding the parentheses for the root.  By
Lemma~\ref{lem:trees}, $B^*$ encodes the balanced-parentheses representation of
a traversal of the spanning tree $T^*$ of the dual of $G$ complementary to $T$
(the right-to-left traversal of $T^*$, in fact), except that it also lacks the
leading 0 and trailing 1 encoding the parentheses for the root.  Therefore,
since $B$ and $B^*$ encode forests, we can support $\match$ and $\parent$ with
them. 

To build $A$, $B$ and $B^*$ given the embedding of $G$ and $T$, we traverse $T$
as in Lemma~\ref{lem:trees}.  Whenever we process an edge, if it is in $T$ then
we append a 1 to $A$ and append the edge to a list $L$; otherwise, we append a 0
to $A$ and append the edge to another list $L^*$.  When we have finished the
traversal, we replace each edge in $L$ or $L^*$ by a 0 if it is the first
occurrence of that edge in that list, and by a 1 if it is the second occurrence;
this turns $L$ and $L^*$ into $B$ and $B^*$, respectively.  For the example
shown in Figure~\ref{fig:trees}, $L$ and $L^*$ eventually contain the edges
shown in the columns labelled $T$ and \(G - T\), respectively, in the table on
the on the right side of the figure, and 
\begin{eqnarray*}
A [1..28] & = & 0110110101110010110100010100 \\
B [1..14] & = & 00101100110011 \\
B^* [1..14] & = & 01001001110101\,.
\end{eqnarray*}

We identify each vertex $v$ in $G$ by its pre-order rank in our traversal of
$T$. We say that, while we visit $v$, we process all the edges that lead from
$v$ to other nodes $w$. Note that each edge $(v,w)$ is processed twice, while
visiting $v$ and while visiting $w$, but these correspond to two distinct
positions in our traversal. Consider the following queries: 

\begin{description}
\item[\(\first (v)\):] return $i$ such that the first edge we process while
  visiting $v$ is the $i$th we process during our traversal; 
\item[\(\last (v)\):] return $i$ such that the last edge we process while
  visiting $v$ is the $i$th we process during our traversal; 
\item[\(\next (i)\):] return $j$ such that if we are visiting $v$ when we
  process the $i$th edge during our traversal, then the next edge we process
when visiting $v$, in counter-clockwise order, is the one we process $j$th; 
\item[\(\prev (i)\):] return $j$ such that if we are visiting $v$ when we
  process the $i$th edge during our traversal, then the previous edge we 
processed when visiting $v$, in counter-clockwise order, is the one we process 
$j$th; 
\item[\(\mate (i)\):] return $j$ such that we process the same edge $i$th and
  $j$th during our traversal; 
\item[\(\vertex (i)\):] return the vertex $v$ such that we are visiting $v$ when
  we process the $i$th edge during our traversal. 
\end{description}
With these it is straightforward to reenact our traversal of $T$ and recover the
embedding of $G$.  For example, with the following queries we can list the edges
incident to the root of $T$ in Figure~\ref{fig:trees} and determine whether they
are in $T$: 
\[\begin{array}{l@{\hspace{3ex}}l@{\hspace{3ex}}l@{\hspace{3ex}}l}
\first (1) = 1 & \mate (1) = 4 & \vertex (4) = 3 & A [1] = 0 \\
\next (1) = 2 & \mate (2) = 10 & \vertex (10) = 2 & A [2] = 1 \\
\next (2) = 11 & \mate (11) = 17 & \vertex (17) = 5 & A [11] = 1 \\
\next (11) = 18 & \mate (18) = 26 & \vertex (26) = 7 & A [18] = 1 \,.
\end{array}\]
To see why we can recover the embedding from the traversal, consider that if we
have already correctly embedded the first $i$ edges processed in the traversal,
then we can embed the \((i + 1)\)th correctly given its endpoints and its rank
in the counter-clockwise order at those vertices. Queries $\last$ and $\prev$ 
are superfluous for this task, but they allow traversing the neighbours of a 
node in clockwise order.

\subsection{Implementing the basic queries}

We now explain our constant-time implementations of
$\first$, $\next$, $\prev$, $\mate$ and $\vertex$.  

\paragraph{Query $\first$}
If \(m = 0\)
then \(\first (v)\) is undefined, which we indicate by returning 0.  Otherwise,
we first process an edge at $v$ immediately after first arriving at $v$.  Since
we identify $v$ with its pre-order rank in our traversal of $T$ and $B$ lacks
the opening parenthesis for the root, while first arriving at any vertex $v$
other than the root we write the \((v - 1)\)th 0 in $B$ and, thus, the
\(B.\select_0 (v - 1)\)th 1 in $A$.  If $v$ is the root then \(\first (v) = 1\)
and so, since \(\select_x (0) = 0\), this case is also handled by the formula
below: 
\[\first (v) = \left\{ \begin{array}{l@{\hspace{1.5ex}}l}
	A.\select_1 (B.\select_0 (v - 1)) + 1 & \mbox{if \(m \geq 1\)} \\
	0 & \mbox{otherwise.}
\end{array} \right.\]
In our example,
\[\first (5)
= A.\select_1 (B.\select_0 (4)) + 1
= A.\select_1 (7) + 1
= 12\,\]
and indeed the twelfth edge we process, \((5, 6)\), is the first one we process
at vertex 5. Note that the formula works for nodes with only one edge too.

\paragraph{Query $\last$}
The logic of $\last$ is similar to that of $\first$; we must locate the 
closing parenthesis that represents $v$ in $T$.

\[\last (v) = \left\{ \begin{array}{l@{\hspace{1.5ex}}l}
 	A.\select_1(B.\match(B.\select_0(v-1))) & \mbox{if \(m \geq 1\)} \\
	0 & \mbox{otherwise.} \\
\end{array} \right.\]

\paragraph{Query $\next$}
If the $i$th edge we process is the last edge we process at a vertex $v$ then
\(\next (i)\) is undefined, which we again indicate by returning 0.  This is the
case when \(i = 2 m\), or \(A [i] = 1\) and \(B [A.\rank_1 (i)] = 1\).
  Otherwise, if the $i$th edge we process is not in $T$, then \(A [i] = 0\), and
  we process the next edge at $v$ one time step later.  Finally, if the
  $i$th edge $e$ we process is in $T$ and not the last one we process at $v$,
  then we next process an edge at $v$ immediately after returning to $v$ by
  processing $e$ again at time \(\mate (i)\).  This is the case when \(A [i] =
  1\) and \(B [A.\rank_1 (i)] = 0\).  In other words, 
\[\next (i) = \left\{ \begin{array}{l@{\hspace{1.5ex}}l}
	i + 1 & \mbox{if $i < 2m$ and \(A [i] = 0\)} \\
	\mate (i) + 1 & \mbox{if $i < 2m$ and \(A [i] = 1\) and \(B [A.\rank_1 (i)] = 0\)} \\
	0 & \mbox{otherwise.}
\end{array} \right.\]
In our example, since \(A [12] = 1\), \(B [A.\rank_1 (12)] = B [8] = 0\), the
twelfth edge we process is \((5, 6)\) and it is also the fifteenth edge we
process, 
\[\next (12)
= \mate (12) + 1
= 16\,,\]
and indeed the second edge we process at vertex 5 is \((5, 7)\).

\paragraph{Query $\prev$}
The logic for $\prev$ is similar to that of $\next$; we only need to consider
that, once we move one position backwards, we might arrive at a closing
parenthesis. The formula follows.
\[\prev (i) = \left\{ \begin{array}{l@{\hspace{1.5ex}}l}
	i - 1 & \mbox{if $i > 1$ and \(A [i-1] = 0\)} \\
	\mate (i-1) & \mbox{if $i > 1$ and \(A [i-1] = 1\) and \(B [A.\rank_1 (i-1)] = 1\)} \\
	0 & \mbox{otherwise.}
\end{array} \right.\]

\paragraph{Query $\mate$}
To implement \(\mate (i)\), we check \(A [i]\) to determine whether we wrote a 
bit in $B$ or in $B^*$ while processing the $i$th edge, and use $\rank$ on $A$
to find that bit in the corresponding sequence. 
We then use $\match$ to find the bit encoding the matching parenthesis, and 
finally use $\select$ on $A$ to find where we wrote in $A$ 
that matching bit.  Therefore, 
\[\mate (i) = \left\{ \begin{array}{l@{\hspace{1.5ex}}l}
	A.\select_0 (B^*.\match (A.\rank_0 (i))) & \mbox{if \(A [i] = 0\)} \\
	A.\select_1 (B.\match (A.\rank_1 (i))) & \mbox{otherwise.}
\end{array} \right.\]
To compute \(\mate (12)\) for our example, since \(A [12] = 1\),
\begin{eqnarray*}
\lefteqn{\mate (12)}\\
& = & A.\select_1 (B.\match (A.\rank_1 (12))) \\
& = & A.\select_1 (B.\match (8)) \\
& = & A.\select_1 (9) \\
& = & 15\,.
\end{eqnarray*}

\paragraph{Query $\vertex$}
Suppose the $i$th edge $e$ we process is not in $T$ and we process it at vertex
$v$.  If the preceding time we processed an edge in $T$ was the first time we
processed that edge, we then wrote a 0 in $B$, encoding the opening parenthesis
for $v$; otherwise, we then wrote a 1 in $B$, encoding the closing parenthesis
for one of $v$'s children.  Now suppose $e$ is in $T$.  If that is the first
time we process $e$, we move to the other endpoint $w$ of $e$ --- which is a
child of $v$ --- and write a 0 in $B$, encoding the opening parenthesis for $w$.
If it is the second time we process $e$, then we write a 1 in $B$, encoding the
closing parenthesis for $v$ itself.  Therefore, 
\[\vertex (i) = \left\{ \begin{array}{l}
	B.\rank_0 (A.\rank_1 (i)) + 1\\\hspace{3ex} \mbox{if \(A [i] = 0\) and
          \(B [A.\rank_1 (i)] = 0\)} \\[1ex] 
	B.\parent (B.\rank_0 (B.\match (A.\rank_1 (i)))) + 1\\\hspace{3ex}
        \mbox{if \(A [i] = 0\) and \(B [A.\rank_1 (i)] = 1\)} \\[1ex] 
	B.\parent (B.\rank_0 (A.\rank_1 (i))) + 1\\\hspace{3ex} \mbox{if \(A [i]
          = 1\) and \(B [A.\rank_1 (i)] = 0\)} \\[1ex] 
	B.\rank_0 (B.\match(A.\rank_1 (i))) + 1\\\hspace{3ex} \mbox{otherwise.}
\end{array} \right.\]
In our example, since \(A [16] = 0\) and \(B [A.\rank_1 (16)] = B [9] = 1\),
\begin{eqnarray*}
\lefteqn{\vertex (16)} \\
& = & B.\parent (B.\rank_0 (B.\match (A.\rank_1 (16)))) + 1 \\
& = & B.\parent (B.\rank_0 (B.\match (9))) + 1 \\
& = & B.\parent (B.\rank_0 (8)) + 1 \\
& = & B.\parent (5) + 1 \\
& = & 5\,,
\end{eqnarray*}
and indeed we process the sixteenth edge \((5, 7)\) while visiting 5.

We remind the reader that since $B$ lacks parentheses for the root of $T$,
\(B.\parent (5)\) refers to the parent of the fifth vertex in an in-order
traversal of $T$ not including the root, i.e., the parent vertex 5 of vertex 6.
Adding 1 includes the root in the traversal, so the final answer correctly
refers to vertex 5.  The lack of parentheses for the root also means that, e.g.,
\(B.\parent (4)\) refers to the parent of vertex 5 and returns 0 because vertex
5 is the root of its own tree in the forest encoded by $B$, without vertex 1.
Adding 1 to that 0 also correctly turns the final value into 1, the in-order
rank of the root.  Of course, we have the option of prepending and appending
bits to $A$, $B$ and $B^*$ to represent the roots of $T$ and $T^*$, but that
slightly confuses the relationship between the positions of the bits and the
time steps at which we process edges. 

We also note that, if we do not require that node identifiers are precisely
preorder ranks in $T$, then we can use the positions of their 0 in $B$ as their
identifiers. This removes the need for using $B.\rank_0$ and $B.\select_0$
in all the formulas that convert between node identifiers and positions in $T$.

\subsection{More complex queries}

We can define more complex queries on top of the basic ones. For
example, we give the pseudocode of three queries: $\degree(v)$ returns the 
number of neighbours of vertex $v$; $\listing(v)$ returns the list of neighbours of 
vertex $v$, in counter-clockwise order; $\face(e)$ returns the list of vertices, in 
clockwise order, of one of the face where the edge $e$ belongs. We also support the
other order (clockwise or counter-clockwise, or the other face where $e$ belongs) by using $\last$ and $\prev$ instead of $\first$ and $\next$.

\begin{figure}[t]
  \centering
  \begin{minipage}[t]{0.31\textwidth}
  \begin{function}[H]
    \footnotesize
    \DontPrintSemicolon
    \SetVlineSkip{0.5ex}
    \LinesNumbered
    \SetKwFunction{Next}{next}
    \SetKwFunction{First}{first}
    \SetKwInOut{Input}{Input}
    \SetKw{Return}{return}
    \Input{node $v$}
    \BlankLine
    $d \asgn 0$\;
    $edg \asgn \First(v)$\;
    \While{$edg \not= 0$} {
      $edg \asgn \Next(edg)$\;
      $d \asgn d + 1$\;
    }
    \Return $d$
    \caption{degree()}
    \label{func:degree}
  \end{function}
\end{minipage}%
\hspace{.3em}%
\begin{minipage}[t]{0.33\textwidth}
  \begin{function}[H]
    \DontPrintSemicolon
    \SetVlineSkip{0.5ex}
    \footnotesize
    \LinesNumbered
    \SetKwFunction{Next}{next}
    \SetKwFunction{First}{first}
    \SetKwFunction{Mate}{mate}
    \SetKwFunction{Vertex}{vertex}
    \SetKwInOut{Input}{Input}
    \SetKw{Out}{output}
    \Input{node $v$}
    \BlankLine
    $edg \asgn \First(v)$\;
    \While{$edg \not= 0$} {
      $mt \asgn \Mate(edg)$\;
      \Out $\Vertex(mt)$\;
      $edg \asgn \Next(edg)$\;
    }
    \caption{listing()}
    \label{func:listing}
  \end{function}
\end{minipage}
\hspace{.3em}%
\begin{minipage}[t]{0.33\textwidth}
  \begin{function}[H]
    \DontPrintSemicolon
    \SetVlineSkip{0.5ex}
    \LinesNumbered
    \footnotesize
    \SetKwFunction{Next}{next}
    \SetKwFunction{First}{first}
    \SetKwFunction{Mate}{mate}
    \SetKwFunction{Vertex}{vertex}
    \SetKwInOut{Input}{Input}
    \SetKw{KwOr}{or}
    \SetKw{KwTrue}{true}
    \SetKw{KwFalse}{false}
    \SetKw{Out}{output}
    \Input{edge $e$}
    \BlankLine
    $edg \asgn e$,
    $\mathit{fst} \asgn \KwTrue$\;
    \While{$edg \neq e$ \KwOr $\mathit{fst}$} {
      $\mathit{fst} \asgn \KwFalse$\;
      $mt \asgn \Mate(edg)$\;
      \Out $\Vertex(mt)$\;
      $edg \asgn \Next(mt)$\;
    }
    \caption{face()}
    \label{func:face}
  \end{function}
\end{minipage}
\end{figure}

Queries $\listing(v)$ and $\face(e)$ are implemented in optimal time, that is, 
$\Oh{1}$ per returned element. Instead, $\degree(v)$ requires time 
$\Oh{\degree(v)}$. We can also determine $\neighbour(u,v)$, that is, whether two
vertices $u$ and $v$ 
are neighbours, by listing the neighbours of each in interleaved form, in time
$\mathcal{O}(\min (\degree (u),$ $\degree (v)))$. These times are not so 
satisfactory compared with the $\Oh{1}$ achieved by other representations 
\cite{MR01,ChiangLinLu2005} to compute $\neighbour(u,v)$ and $\degree(v)$.

For $\degree(v)$, we can get arbitrarily close to constant time by adding $o(m)$
further bits to our representation, that is, we can solve the query in time 
$\Oh{f(m)}$
for any given function \(f(m) \in \omega (1)\). To do this, we store a bitvector
$D[1..n]$ marking with 1s the (at most) \(m/f(m) = o(m)\) vertices with degree 
at least \(f (m)\), which takes $nH(m/(nf(m)))+o(n)=\Oh{(m/f(m))\log (nf(m)/m)}
+o(n)= o(m)$ bits by using a sparse bitvector representation \cite{RRR07} 
(recall that $G$ is connected, so $m \ge n-1$). We also store a second 
bitvector $E[1..2m]$ where we append, for each $D[v]=1$, 
$\degree(v)-1$ copies of 0s followed by a 1. Since $E$ has $m/f(m)$ 1s,
it can also be stored as a sparse bitvector using $\Oh{(m/f(m))\log f(m)}+o(m)
= o(m)$ bits. Therefore, if $D[v]=1$, its degree is obtained in constant
time with $\select_1(E,r)-\select_1(E,r-1)$, where $r=\rank_1(D,v)$. If, 
instead, $D[v]=0$, then we know that $\degree(v) < f(m)$ and thus we apply the
procedure that sequentially counts the neighbours, in time $\Oh{f(m)}$.

We can use a similar idea, albeit more complex, to answer $\neighbour(u,v)$ 
queries in time $\Oh{f(m)}$, for any $f(m) = \omega(\log m)$.
We consider the graph induced by the $\Oh{m/f(m)} = o(m/\log m) $ nodes with
degree $f(m)$ or higher and
eliminate multi-edges and self-loops. The resulting graph $G'$ is simple and 
still planar, so it has average degree less than 6 and thus $o(m/\log m)$
edges. We represent $G'$ in classical adjacency-list form, with the nodes 
inside each list sorted by increasing node identifier. This requires $o(m)$ 
bits in total. To 
solve $\neighbour(u,v)$ in $G'$, we can use binary search for $v$ in the list of $u$
in time $\Oh{\log m}= o(f(m))$.
To answer \(\neighbour (u, v)\) on $G$, we check whether either $u$ or $v$ is
low-degree (assuming we mark low-degree nodes in a bitvector $D'$ analogous to
$D$) and, if so, list its neighbours in $\Oh{f(m)}$ time. If not, we translate
nodes $u$ and $v$ to their corresponding nodes $u'=\rank_1(D',u)$ and
$v'=\rank_1(D',v)$ in $G'$ and query $G'$ in time $o(f(m))$.

The following theorem summarizes the results of this section.

\begin{theorem}
\label{thm:main}
We can store a given planar embedding of a connected planar graph $G$ with $m$
edges in \(4 m + o (m)\) bits such that later, given a vertex $v$, we can list 
the edges incident to $v$ in clockwise or counter-clockwise order, even if we 
are given a particular starting edge incident to $v$, using constant time per 
edge. We can also traverse the edges limiting a face in constant time per 
edge. Further, we can find a vertex's degree in $\Oh{f (m)}$ time for any
given function \(f (m) \in \omega (1)\), and determine whether two
vertices are neighbours in $\Oh{f (m)}$ time for any
given function \(f (m) \in \omega (\log m)\).
\end{theorem}

\subsection{Reducing space on simple planar graphs}

Chiang {\it et al.}~\cite{ChiangLinLu2005} use $2m+3n+o(m)$ bits to represent 
planar graphs without self-loops, which can be more than the $4m+o(m)$ bits 
used in our representation. However, if $G$ is simple (i.e., has no loops nor
multiple edges), their representation requires only $2m+2n+o(m) \le 4m+o(m)$ 
bits. We remind the reader that this representation can handle any simple planar graph,
but does not always respect the given embedding, so they cannot represent
arbitrary embeddings.

We show that, if there are no self-loops, our representation can use less 
than $4m+o(m)$ bits, by exploiting some redundancy in our representation
and without changing the main scheme. Assume we represent a single sequence 
$S[1..2m]$ over an alphabet of four symbols, $\Sigma = \{ (, ), [, ] \}$, 
that replaces $A$, $B$, and $B^*$.
That is, the parentheses are the 0s and 1s in $B$, the brackets are the 0s and
1s and $B^*$, and $A$ corresponds to whether the symbols are parentheses or
brackets. In our running example, the sequence is
\[ S[1..2m] ~~=~~ [~(~(~]~)~(~[~)~[~)~(~(~]~[~)~[~)~(~]~(~]~]~[~)~]~)~[~].
\]
The zeroth-order entropy of $S$ is defined as 
$H_0(S) = \sum_{c \in \Sigma} \frac{m_c}{2m}\log \frac{2m}{m_c}$, where $c$ 
occurs $m_c$ times in $S$. The $k$th-order entropy, for any $k>0$, is defined as
$H_k(S) = \sum_{C \in \Sigma^k} \frac{|S_C|}{2m} H_0(S_C)$, where $S_C$ is the
string formed by the symbols that follow the context $C$ in $S$ (assume $S$ 
is circular for simplicity, so that $S[1]$ follows $S[2m]$).

Ferragina and Venturini \cite{FV07} show how to store a string $S$ within
$|S|H_k(S) + o(|S|\log|\Sigma|)$ bits, for any $k=o(\log_{|\Sigma|} |S|)$,
so that any substring of length $\Oh{\log |S|}$ can be extracted in constant 
time. We use their result to store $S$ in $2mH_1(S)+o(m)$ bits. Instead of a 
structure on parentheses on bitvector $B$ and another on bitvector $B^*$, we 
build both parentheses structures on top of sequence $S$. Both are similar to
the original $o(m)$-bit structure of Navarro and Sadakane \cite{NS14}, only that
the structure built to navigate parentheses ignores the bracket symbols,
and vice versa (a similar arrangement is described by 
Navarro~\cite[pp.~311--315]{Navarro2016}). The only
changes are that each symbol uses 2 bits instead of 1, that there are two 
symbols that do not change the ``excess'' count (number of opening minus
closing parentheses up to some position), and that in order to extract a chunk
of $\Theta(\log m)$ symbols, we use the extraction method of Ferragina and 
Venturini \cite{FV07}. A $\rank$/$\select$ functionality on top of $A$ is also
easily provided on top of $S$, by using the same $o(m)$-bit structures 
\cite{Cla96,Mun96} and interpreting both parentheses as 1s and both brackets 
as 0s.
Therefore, with $o(m)$ further bits, we provide the necessary functionality on
top of the $H_1(S)$ bits needed to encode $S$.

This entropy gives precisely 2 bits per symbol (and thus $4m$ bits in total) 
for general planar embeddings, but if there are no self-loops, then the 
substring ``[~]'' cannot appear in $S$ (other longer strings cannot appear
either, but we would need a higher-entropy model to capture them). An upper
bound to the first-order entropy when this substring is forbidden is obtained
by noticing that we can have only 3 symbols, instead of 4, following an
opening bracket; therefore we can encode $S$ using $n\lg 4 + n\lg 4 +
(m-n)\lg 3 + (m-n)\lg 4 = m \lg 12 + n \lg(4/3) \approx 3.58 m + 0.42 n$. This
is still $4m$ in the worst case.
To obtain a nontrivial bound in terms of $m$, we calculate the exact
first-order entropy of $S$ when substring ``[~]'' is forbidden. 

Let us use the names $op=($,
$cp=)$, $ob=[$, and $cb=]$. Let us call $x_y$ number of symbols $y$
following a symbol $x$ in $S$; for example $op_{ob}$ is the number of opening
brackets following opening parentheses, that is, the number of occurrences of
substring ``(~['' in $S$. It must then hold that 
$\sum op_* = \sum cp_* = n$ and $\sum ob_* = \sum cb_* = m-n$. It also holds
$\sum *_{op} = \sum *_{cp} = n$ and $\sum *_{ob} = \sum *_{cb} = m-n$.
The system of restrictions must be satisfied while maximizing
$$2m H_1(S) ~=~ n H(op_*) + n H(cp_*) + (m-n) H(ob_*) + (m-n) H(cb_*),$$
where
$H(x_1,\ldots,x_4) = \sum \frac{x_i}{x}\log\frac{x}{x_i}$ and $x=\sum x_i$.
Forbidding self-loops implies the additional restriction $ob_{cb} = 0$.

We solve the optimization problem with a combination of algebraic and numeric 
computation, using Maple and C, up to 4 significant digits. We find that the 
entropy is maximized at a value slightly below $3.8 m$.%
\footnote{The maximum is $3.7999 m$, found for $m = 1.731 n$, 
$op_{op} = cp_{op} = 0.2683 n$, $ob_{op} = 0.2679 n$, $cb_{op} = 0.1955 n$, 
$op_{cp} = cp_{cp} = 0.2677 n$, $ob_{cp} = 0.2673 n$, $cb_{cp} = 0.1974 n$, 
$op_{ob} = cp_{ob} = 0.1961 n$, $ob_{ob} = 0.1958 n$, $cb_{ob} = 0.1429 n$, 
$op_{cb} = cp_{cb} = 0.2679 n$, $ob_{cb} = 0, cb_{cb} = 0.1952 n$.}
Therefore, the resulting space with no self-loops and using the described
compressed representation can be bounded by $3.8 m + o(m)$ bits.
Simple graphs have no self-loops and no multiple-edges, but this second
restriction translates into longer forbidden substrings, whose effect is harder 
to analyze.

We remind the reader that the representation of Keeler and Westbrook 
\cite{KeelerWestbrook1995}, on the other hand, achieves $m\lg 12 \approx 3.58 m$
bits when no self-loops (or, alternatively, no degree-one nodes) are permitted,
yet it does not support queries. When neither self-loops nor degree-one nodes are permitted,
they reach $3m$ bits. In this case, both ``[~]'' and ``(~)'' are forbidden
strings. While we have not been able to compute the exact first-order entropy
in this case, this must be at most $n \lg 3 + n \lg 4 + (m-n)\lg 3 + (m-n)\lg 4
= m\lg 12 \approx 3.58m$, which is obtained by using $\lg 4$ bits to encode
the symbol that follows a closing bracket or parenthesis, and $\lg 3$ bits to
encode the symbol that follows an opening bracket or parenthesis.

We note that these space improvements can also be applied on top of the
representation of Chiang {\it et al.}~\cite{ChiangLinLu2005} since, when 
encoding a simple graph, the difference between both representations is that
they use a particular spanning tree (which may also force a particular embedding).


\subsection{Unconnected planar graphs}

Our representation can be easily extended to unconnected planar graphs, because
our parentheses representations can immediately be extended to handle forests
instead of just individual trees. To handle an unconnected planar graph, 
we first find
all the connected components of the graph and then compute an
arbitrary spanning tree for each connected component. Then, we
construct the binary sequences: the sequence $B$ will represent the
{\em forest} of the spanning trees, concatenating all the
balanced-parentheses representations; the sequence $B^*$ will
represent the complementary spanning tree of the dual of the graph.
Finally, sequence $A$ indicates the interleaving of the sequences $B$ 
and $B^*$. We visit the connected components in arbitrary order.

Note that, in the case of connected planar graphs, our navigation queries and
the fact that the first edge we list is adjacent to the external face, are
sufficient to recover the embedding. This is not the case if the graph has
$k>1$ connected components. Concretely, some components may be embedded 
inside faces of other components, whereas our arrangement assumes that all the
connected components lie on the outer face (our navigation queries cannot
distinguish between those cases). 

To recover the embedding we might add $k-1$ edges to the spanning tree, so
that all the connected components lying in a single face are threaded through
a node in their frontier, and the first one is linked to a node on the 
containing face. Therefore the total length of $A$
will be $2(m+k-1)$ and the length of $B$ will be $2(n+k-2)$. The fake edges
will be marked in a bitvector $K[1..n+k-2]$ indexed by preorder value. Since 
$K$ contains $k-1$ 1s, it can be encoded in 
$k\lg(n/k)+\Oh{k}$ bits \cite{OS07}. Since $n \le m+k$,
the space of the whole structure can be written in terms of $m$ and $k$ 
as $4m+k\lg(m/k)+\Oh{k}+o(m)$ bits. 

The $k\lg(m/k)+\Oh{k}$ or $k\lg(n/k)+\Oh{k}$
bits to describe the embedding are asymptotically 
optimal: consider a chain of $t$ triangles (delimited with $m=2t+1$ edges)
and $k-1$ isolated nodes (so there are $k$ connected components in total) to 
represent all the ways to distribute $k-1$ balls into $t$ bins. This requires 
$\lg { k+t-2 \choose k-1} = k\lg(t/k) + \Oh{k} = k\lg(m/k) + \Oh{k}$ bits with 
any encoding. This is also $k\lg(n/k)+\Oh{k}$ bits, since this graph has 
$n=2t+k$ nodes.

We use $K$ to avoid listing fake edges in any of the traversal operations.
The fake edges increase the degree of a node by a constant factor: a node may
have one fake edge per face it participates in, which at most doubles its
degree. Further, a node in the frontier of its component may have two extra
fake edges threading it with other connected components. Therefore, the time
complexity of the navigation operations is not affected.

The fake edges may, in addition, be useful for a more ambitious $\face$
operation that takes into account the actual embedding, where a face is 
surrounded by a sequence of edges but is also limited by the frontier edges 
of the connected components it has inside. To find all those edges, we also
traverse the fake edges in the $\face$ traversal, yet without listing them. 
The fake edges will lead us to the other connected components that are 
contained and/or surround the face we are listing.

\section{Parallel construction}
\label{sec:parallel}

In this section we discuss the parallel construction of our extension of 
Tur\'an's representation. Since the representation is based on spanning trees 
and tree traversals, we can borrow ideas of well-known parallel algorithms, 
such as parallel Euler Tour traversal or parallel computation of spanning trees.

We assume that a tree $T$ is represented with adjacency lists. Such 
representation consists of an array of nodes $V_{T}[1..n]$, and an array of 
edges $E_{T}[1..2n-2]$. Each node $v\in V_{T}$ stores two indices in $E_{T}$, 
$v.\mathit{first}$ and $v.\mathit{last}$, delimiting the adjacency list of $v$,
which starts with $v$'s parent edge (except the root) and is sorted 
counter-clockwise around $v$. The number of children of $v$ is then 
$v.\mathit{last}-v.\mathit{first}$ (plus 1 for the root). Each edge 
$e\in E_{T}$ has three fields: $e.src$ and $e.tgt$ are the positions
in $V_T$ of its source and target vertices, and $e.mat$ is the position in 
$E_{T}$ of the mate edge $e'$ of $e$, where $e'.src = e.tgt$ and $e'.tgt = 
e.src$. Our representation of
graphs is similar, with the exception that the concept of \emph{parent} of a
vertex is not valid in graphs; therefore the first edge in the adjacency list of
a vertex $v$ cannot be interpreted as $v$'s parent edge.

\subsection{Parallel construction of compact planar embeddings}
\label{subsuc:paralgo}

We will first assume that the input consists of a connected planar graph 
embedding $G=(V_G,E_G)$ and a spanning tree $T=(V_T,E_T)$ of $G$, together with
an array $C$ that stores the number of edges of $G\setminus T$ between any 
two consecutive edges in $T$, in counter-clockwise order. In Section~\ref{subsec:spanning} we will explain how to obtain $T$ and $C$ in parallel. 

With the spanning tree, we construct the bitvectors $A$, $B$, and $B^*$ by 
performing an Euler Tour over $T$. During the tour, by writing a
$0$ for each forward (parent to child) edge and a $1$ for each backward (child
to parent) edge, we obtain the bitvector $B$. By reading in $C$ the number of 
edges of $G\setminus T$ between two consecutive edges of $T$, representing
these edges with 0s and the edges of $T$ with 1s, we obtain the bitvector $A$. 
Finally, by using the previous Euler Tour and the array $C$ we can obtain the 
bitvector $B^*$, by finding out which is the first (0) and which is the second
(1) occurrence of each edge. 

Algorithm \ref{algo:parAlgo} gives the detailed pseudocode. It works in the following steps:

\begin{algorithm2e}[t!]
  \footnotesize
  \SetKwInOut{Input}{Input}
  \SetKwInOut{Output}{Output}
  \SetKwFor{PFor}{parfor}{do}{end}
  \SetKwFunction{parListRanking}{parallelListRanking}
  \SetKwFunction{createRS}{createRankSelect}
  \SetKwFunction{createBP}{createBP}
  \LinesNumbered
  \DontPrintSemicolon
  \SetVlineSkip{0.5ex}
  \SetCommentSty{textit}
    \SetKw{KwOr}{or}
  \Input{A planar graph embedding $G=(V_{G},E_{G})$, a spanning
    tree $T=(V_{T},E_{T})$ of $G$, an array $C$ of
    size $|E_T|$, and the starting vertex $\mathit{init}$.}
  \Output{Bitvectors $A$, $B$ and $B^*$ induced by $G$ and $T$.}
  \BlankLine
  $\mathit{A} \asgn {}$a bitvector of length $|E_{G}|$ initialized with 0s\;
  $\mathit{B} \asgn {}$a bitvector of length $|E_{T}|-2$\;
  $\mathit{B^*} \asgn {}$a bitvector of length $|E_{G}|-|E_{T}|+2$ initialized with 0s\;
  $\mathit{LE} \asgn {}$an array of length $|E_{T}|$\;

  \PFor{$j \asgn 1$ \KwTo $|E_T|$}{

      $\mathit{LE}[j].\mathit{rankA} \asgn C[E_{T}[j].mat]+1$\;
      $\mathit{LE}[j].\mathit{rankB} \asgn 1$\;
      
      \eIf(\tcp*[h]{forward edge}){$E_{T}[j].src=\mathit{init}$ \KwOr
        $E_{T}[j].src.\mathit{first}\neq j$}{
            $\mathit{LE}[j].\mathit{value} \asgn 0$ \tcp*[h]{opening parenthesis}

            \eIf(\tcp*[h]{target is a leaf}){$E_{T}[j].tgt.\mathit{first}=E_{T}[j].tgt.\mathit{last}$} {
                   $\mathit{LE}[j].\mathit{succ} \asgn E_{T}[j].mat$\;
            }(\tcp*[h]{target has children})
            {
                   $\mathit{LE}[j].\mathit{succ} \asgn E_{T}[j].tgt.\mathit{first}+1$\;
            }
      }(\tcp*[h]{backward edge})
      {
            $\mathit{LE}[j].\mathit{value} \asgn 1$ \tcp*[h]{closing parenthesis}

            \eIf(\tcp*[h]{$j$ was the last edge of target, return}){$E_{T}[j].mat = E_T[j].tgt.last$}{ 
                 $\mathit{LE}[j].\mathit{succ} \asgn E_{T}[j].tgt.\mathit{first}$\;
            }(\tcp*[h]{continue with next edge from target})
            {
                $\mathit{LE}[j].\mathit{succ} \asgn E_{T}[j].mat+1$\;
            }
       }
    }
  $\parListRanking(\mathit{LE})$\;
  
  \PFor{$j \asgn 1$ \KwTo $|E_T|$} {
      $A[\mathit{LE}[j].\mathit{rankA}] \asgn 1$ \;
      $B[\mathit{LE}[j].\mathit{rankB}] \asgn \mathit{LE}[j].\mathit{value}$\;
    }

  $\mathit{D}_{pos}, \mathit{D}_{edge} \asgn {}$arrays of length
  $|E_G|$ and $|E_{G}|-|E_{T}|+2$, respectively\;
  \PFor{$j \asgn 1$ \KwTo $|E_T|$} {
      $p \asgn \mathit{LE}[j].\mathit{rankA}-\mathit{LE}[j].\mathit{rankB}$\;
      $base = \mathit{ref}(E_{T}[j].mat)$\;
      $delta \asgn p - base -1$\;
      \PFor{$k \asgn base+1$ \KwTo $base+C[E_{T}[j].mat]$} {
        $\mathit{D}_{pos}[k] \asgn k+delta$\;
        $\mathit{D}_{edge}[k+delta] \asgn k$\;
      }
    }
  \PFor{$j \asgn 1$ \KwTo $|E_{G}|-|E_{T}|+2$} {
      $mat \asgn E_{G}[\mathit{D}_{edge}[j]].mat$\;
      \If{$j > \mathit{D}_{pos}[mat]$} {
        $B^{*}[j] \asgn 1$ \;
      }
    }
  $\createRS(\mathit{A})$, $\createBP(\mathit{B})$, $\createBP(\mathit{B^*})$\;
  \caption{Parallel compact planar embedding algorithm.}
  \label{algo:parAlgo}
\end{algorithm2e}

\begin{enumerate}
\item
In lines 1--4, it initializes the output bitvectors ($A$ and $B^*$ are set to
0s) and creates an
auxiliar array $LE$ that is used to store the traversal of the tree
following the Euler Tour. Each entry of $LE$ represents one traversed edge of
$T$ and stores four fields: {\em value} is $0$ or $1$ depending on whether the
edge is a forward or a backward edge, respectively; {\em succ} is the index in
$LE$ of the next edge in the Euler tour; {\em rankA} is the rank of the edge in
$A$; and {\em rankB} is the rank of the edge in $B$. 
\item
In lines 5--19, the algorithm traverses $T$ to create the Euler Tour.
For each edge $e_j\in E_T$, {\em rankA} is set to $C[E_{T}[j].mat]+1$ and 
{\em rankB} to $1$ (lines 6--7). Those ranks will be used later to compute the 
final positions of the edges in $A$, $B$, and $B^*$. For
each forward edge, a $0$ is written in the corresponding {\em value} field and
the {\em succ} field is connected to the next edge in the Euler Tour. For
backward edges the procedure is similar. Note that all the edges in the 
adjacency list of a node of $T$ are forward edges, except (for non-root nodes)
the first one, which is the parent edge. 
\item
Line 20 computes the final ranks in $A$ and $B$ using a parallel list ranking 
algorithm that adds up the weights from the beginning of the list to each
element. The weights are stored in the fields {\em rankA} and {\em rankB} 
of $LE$. We use the list ranking algorithm of Helman and J\'aj\'a
\cite{Helman2001265}.
\item
Bitvectors $A$ and $B$ are written in lines 21--23. Since
initially all the elements of $A$ are $0$s, it is enough to set to $1$ all the
elements in the {\em rankA} fields. For $B$, the algorithm copies the
content of field {\em value} at position {\em rankB}, for all the elements in
$LE$.
\item
The algorithm now computes the position of each edge of $G\setminus T$ in $B^*$.
That information is implicit in the fields {\em rankA} and {\em rankB} of $LE$ 
(line 26), once the list ranking of step 3 is carried out. For each edge 
$e\in E_T$, the algorithm computes the positions in $B^*$ of the edges of 
$G\setminus T$ that follow, in counter-clockwise order, the mate edge of $e$ 
(lines 27--31). The algorithm uses two auxiliary arrays, $\mathit{D}_{pos}$ and 
$\mathit{D}_{edge}$. Let edge $E_G[j]$ belong to $G\setminus T$. Then 
$\mathit{D}_{pos}[j]$ stores the position of the edge in $B^*$. The array 
$\mathit{D}_{edge}$ is the inverse of $\mathit{D}_{pos}$: $\mathit{D}_{edge}[i]$
is the position of the $i$-th edge of $B^*$ in $E_G$. This step uses function
$\mathit{ref}(e)$, which maps the position $e$ of an edge in $E_T$ to its 
position in $E_G$. This is naturally returned by the spanning tree 
construction, which gives the identity in $G$ of the edges selected for $T$.
\item
In lines 32--35, the algorithm computes whether the edges stored in
$\mathit{D}_{pos}$ are forward or backward edges. For each edge $e$ in 
$G\setminus T$, it compares the positions in $B^*$ of $e$ and
its mate. If the position of $e$ is greater, then $e$ is a backward edge and, 
therefore, is represented with a $1$.
\item
Finally, in line 36 the structures to support operations $\rank$, $\select$, 
$\match$, and $\parent$ are constructed. For the bitvector $A$, the parallel 
algorithm of Labeit {\it et al.}~\cite{LSB17} ($\mathtt{createRankSelect}$) is
used. For $B$ and $B^*$ the parallel algorithm of Ferres {\it et
  al.}~\cite{Fuentes-Sepulveda2017} for balanced parenthesis sequences
($\mathtt{createBP}$) is used.
\end{enumerate}

We have omitted some implementation details for simplicity. For example, the
pseudocode
uses {\bf parfor} throughout, whereas the implementation uses the threads
in a more controlled manner. Line 29, in particular, is more efficiently done
in sequential form. We have also omitted some space optimizations, such as
the reuse of some fields instead of allocating new arrays.

\paragraph{Analysis}

Step 1 initializes the arrays, which requires $T_1=\Oh{m}$ work and 
$T_\infty = \Oh{\log m}$ span (due to the overhead of the implicit {\bf 
parfor}). In step 2, the algorithm traverses the edges of $T$, performing an 
independent computation on each edge. Therefore, with the overhead of the 
{\bf parfor} loop, we obtain $T_1=\Oh{n}$ and $T_\infty=\Oh{\lg n}$ time. Step
3 uses a parallel list ranking algorithm \cite{Helman2001265} over $n$ elements,
which has complexities $T_1=\Oh{n}$ and $T_\infty=\Oh{\lg n}$. Step 4 assigns
the values to $A$ and $B$ independently for each entry, thus we have again
$T_1=\Oh{n}$ and $T_\infty=\Oh{\lg n}$. In step 5, the algorithm traverses all 
the edges in $G\setminus T$. Since the loop in line 29 is also processed in 
parallel, we obtain $T_1=\Oh{m-n}$ and $T_\infty=\Oh{\lg(m-n)}$. 
Similarly to step 4, in step 6 the algorithm sets the entries of 
bitvector $B^*$, which can be done independently for each entry, obtaining times
$T_1=\Oh{m-n}$ and $T_\infty=\Oh{\lg(m-n)}$. Finally, step 7 builds the
rank/select structures in times $T_1=\Oh{m}$ and $T_\infty=\Oh{\lg m}$ 
\cite{LSB17}. The construction of the structures supporting $\match$ and 
$\parent$ over balanced parentheses is constructed in times $T_1=\Oh{m}$ and 
$T_\infty=\Oh{\lg m}$ \cite{Fuentes-Sepulveda2017}.

In addition to the size of the compact data structure, our algorithm uses
$\Oh{m\log m}$ bits for the arrays $LE$, $D_{pos}$ and $D_{edge}$. As said, the
constant is kept low in practice by reusing fields. Notice that the
memory consumption is independent of the number of threads.

\subsection{Structures for degree and neighbour queries}
\label{sec:extrastructs}

Before discussing how to construct the structures to speed up $\degree(v)$ and
$\neighbour(u,v)$ queries, let us discuss the parallel construction of
the sparse bitvector of Raman {\it et al.}~\cite{RRR07}. Let $\ell$ be the
length of the sparse bitvector. Their representation divides the 
bitvector into blocks of length $b = (\log \ell)/2$.
The $i$th block is described as a pair $(c_i,o_i)$, where
$c_i$ corresponds to the number of $1$s inside the block, also known as the {\em
class} of the block, and $o_i$ corresponds to its {\em offset}, an identifier
among all the different blocks sharing the same class. Thus, the bitvector is 
represented as two arrays, $C[1..\lceil \ell/b\rceil]$ and $O[1..\lceil
  \ell/b\rceil]$, where $C[i]=c_i$ and $O[i]=o_i$. We can compute in parallel
each entry of the arrays $C$ and $O$ independently, using linear time
on each block \cite[Sec.~4.1]{Navarro2016}. Thus, we have $\Oh{\ell}$ work and 
$\Oh{\lg (\ell/b)+b}=\Oh{\lg \ell}$ span. In order
to reduce the space consumption of the arrays $C$ and $O$, the entries of the 
arrays are packed into the bits of consecutive machine words. Notice that the 
size of the elements of $C$ is fixed, $\lceil \lg(b+1)\rceil$ bits, whereas
the size of those of $O$, $\lceil \lg o_i\rceil$ bits, is variable. To pack 
the entries of $O$ in parallel, 
we need to compute an array $P[1..\lceil \ell/b\rceil]$ pointing to the
starting position of each element in $O$. Array $P$ is computed with a
parallel parallel prefix sum over the values $\lceil \lg o_i\rceil$.
This takes linear work and logarithmic span \cite{LSB17}, and then we
can write each value $o_i$ to its packed position in parallel. The array $P$
is retained to provide efficient access to $O$. To reduce its space to
$o(\ell)$ bits, only the entries of the form $P[i \cdot \log n]$ are stored in 
absolute form, whereas the others are stored as differences from the preceding 
multiple of $\log n$, using $\Oh{\log\log n}$ bits. This space reduction is
easily computed in parallel within the same time bounds. 
Once the data structures $C$, $O$, and $P$, using 
$\ell H + o(\ell)$ bits, are built, we can access in constant time any chunk of
$\Oh{\log \ell}$ bits from the bitvector by using tables \cite{RRR07}. Therefore,
we can provide $\rank$ and $\select$
functionality by building the classical $o(\ell)$-bit data structures on top
of the bitvector, in parallel \cite{LSB17}. In total, we use $\Oh{\ell}$ work
and $\Oh{\log \ell}$ span.

The structures to support $\degree(v)$ can then be constructed in parallel as
follows: First, we construct the bitvector $D$ by checking all the vertices with
degree at least \(f (m)\). Remember that the degree of a vertex $v$ can
be computed in constant time with $v.\mathit{last}-v.\mathit{first}$. Since the
degree of each vertex can be obtained independently, we can do this in parallel
with $\Oh{m}$ work and $\Oh{\lg m}$ span. Then, we construct the bitvector $E$ by
writing in unary the degree of each high-degree vertex. To do that, we perform a
parallel prefix sum over all the degrees of high-degree vertices. The prefix sum
returns the positions where we have to write a $1$ in $E$. Thus, we construct
$E$ with $\Oh{m}$ work and $\Oh{\lg m}$ span. Finally, we construct the compact
representation of $D$ and $E$ in $\Oh{m}$ work and $\Oh{\lg m}$ span, using the
sparse bitvectors of Raman {\em et al.}~\cite{RRR07}.

For the $\neighbour(u,v)$ query, we must contract the original graph $G$ into
a smaller graph $G' =(V',E')$, induced by all the vertices with degree at least
\(f (m)\). To build $G'$ efficiently in parallel we do as follows. We first
compute $D'[1..n]$ similarly to $D$. We then fill two arrays $X[1..n]$ and 
$Y[1..2m]$, so that $X[i]=D'[i]$; and $Y[j]=1$ if $D'[E_G[j].src]=1$ and 
$D'[E_G[j].tgt]=1$, and $Y[j]=0$ otherwise. Next, we perform a parallel prefix 
sum over $X$, so that $X[i]$ is the name of node $i$ in $G'$ (if $D'[i]=1$). 
We also perform a parallel prefix sum on $Y$, so as to write contiguously in 
array $E'$ the mapped edge targets, $E'[j']=X[E_G[j].tgt]$ for those entries 
$j$ where $Y[j]=1$, where $j'=\sum_{k=1}^j Y[k]$. For each such edge, we also 
check if it 
is the first with this $X[E_G[j].src]$ value, and if so, we record that $j'$ 
is the start of the adjacency list of node $X[E_G[j].src]$, in an array 
$V'[X[E_G[j].src]] = j'$.

Thus $V'$ and $E'$ are an adjacency list representation of $G'$, built with 
$\Oh{m}$ work and $\Oh{\lg m}$ span. Instead of sorting the adjacency lists, 
however, we build a wavelet tree representation on $E'$ \cite{LSB17}. This 
supports the operation $\rank$ generalized to sequences, and therefore we use
that high-degree nodes $u$ and $v$ of $G$ are connected if and only if $X[v]$ is mentioned
in the adjacency list of $X[u]$, that is,
$E'.\rank_{X[v]}(V'[X[u]+1]-1)-E'.\rank_{X[v]}(V'[X[u]]-1)>0$.
The generalized $\rank$ operation takes time $\Oh{\log |V'|}$ and the wavelet
tree is built with $\Oh{|E'|}=o(m/\log m)$ work and $\Oh{\log^2 |E'|} =
\Oh{\log^2 m}$ span.

\begin{lemma}
  \label{lem:parallel}
  Given a connected planar graph embedding $G$ with $m$ edges and a spanning 
  tree of $G$, we can compute in parallel a compact representation of $G$, 
  using $4m + o(m)$ bits and supporting the navigational operations 
  described in Section~\ref{sec:structure}, in $\Oh{m}$ work and $\Oh{\lg m}$ 
  span ($\Oh{\lg^2 m}$ span if operation $\neighbour$ is supported), using 
  $\Oh{m\lg m}$ bits of additional memory.
\end{lemma}

\subsection{Parallel computation of spanning trees}
\label{subsec:spanning}

In this section we discuss the parallel computation of the spanning tree 
$T=(V_T,E_T)$ and the array $C$ used in Section~\ref{subsuc:paralgo}.

Generating a rooted (or a directed) spanning tree turns out to be a difficult
to parallelize problem. Even if it seems to be easier on planar embeddings,
we do not know of good worst-case results on the DyM model. We discuss practical
solutions later.

Such a spanning tree algorithm returns an array of parent references for each 
vertex. With this array of references, we can construct the corresponding 
adjacency list
representation of the spanning tree. To do that, we mark with a
$1$ each edge $E_G$ that belongs to $E_T$ and with a $0$ the rest of the
edges. Using a parallel prefix sum algorithm over $E_G$, we compute the position
of all the marked edges of $E_G$ in $E_T$. The $\mathit{first}$ and
$\mathit{last}$ fields of each node in the spanning tree are computed
similarly. As a byproduct of the computation of $E_T$, we can compute the array
$C$, which stores the number of edges of $G\setminus T$ between two consecutive
edges in $T$, in counter-clockwise order. This can be done by using the marks in
the edges, counting the number of $0$s between two consecutive $1$s. Note that
the starting vertex for the spanning tree must be in the outer face of $G$,
to meet the description of the compact data structure for planar embeddings.
Overall, we require times $T_1=\Oh{m}$ and $T_\infty=\Oh{\lg m}$ once the
spanning tree is built, which is the complexity of the variants of the
parallel prefix sum algorithm we employ. By combining the results with
Lemma~\ref{lem:parallel}, we have the main result on construction.

\begin{theorem}
\label{thm:parallel}
The compact representation introduced in Theorem~\ref{thm:main} of a 
connected planar graph embedding $G$ with $m$ edges
can be constructed under the Dynamic Multithreaded parallel model with 
$\Oh{m+\mathrm{spw}}$ work and $\Oh{\log m+\mathrm{sps}}$ span
($\Oh{\log^2 m+\mathrm{sps}}$ span if operation $\neighbour$ is supported), 
where $\mathrm{spw}$ and $\mathrm{sps}$ are the work and span, respectively, of 
any rooted spanning tree algorithm on planar embeddings.
\end{theorem}

\paragraph{In practice}
The generation of a spanning tree is also difficult to parallelize in
practice. Bader and Cong \cite{BaderCong2005} mention that ``the spanning tree 
problem is notoriously hard for any parallel implementation to achieve 
reasonable speedup'', and propose an algorithm that is shown to perform well
in practice. This is the one we use in our implementation.

Their algorithm works as follows. Given a starting vertex of the graph $G$ 
with $n$ vertices and $m$ edges, the algorithm computes
sequentially a spanning tree of size $\Oh{p}$, called {\em stub spanning
  tree}, where $p$ is the number of available threads. Then, the leaves of the 
stub spanning tree are evenly assigned 
to the $p$ threads as starting vertices. Each thread traverses $G$, using its
starting vertices, constructing spanning trees with a DFS traversal using a
stack. For each vertex, a reference to its parent is assigned. Since a vertex can be visited by
several threads, the assigment of the parent of the vertex may genarate a {\em race
  condition}. However, since the parent assigned by any thread already
belongs to a spanning tree, any assignment will generate a correct tree. Thus,
the race condition is benign. Once a thread has no more
vertices on its stack, it tries to steal vertices from the stack of another
thread by using the work-stealing algorithm. Since the spanning
trees generated by all the threads are connected to the stub spanning tree, the
union of all the spanning tree generates a spanning tree of $G$. 

They analyze their algorithm in expectation on random graphs, obtaining
$\Oh{m/p}$ time when $p \ll m$, but general random graphs have a very small 
diameter. The diameter seems to be a lower bound for the span of their 
algorithm, and this is $\Theta(n^{1/4})$ on random planar graphs \cite{CFGN15}.
Also, their best possible time is $\Oh{\sqrt{m}}$, achieved when using 
$p=\sqrt{m}$ processors. Despite its analysis, the algorithm of Bader
and Cong has a good practical behavior and its implementation is simple. 

To handle unconnected planar graphs, we can first use the algorithm of
Shun {\it et al.}~\cite{Shun:2014:SPL:2612669.2612692}, which finds the
connected components within $\Oh{n}$ work and $\Oh{\log^3 n}$ span with high 
probability, and is shown to perform well in practice.

\paragraph{PRAM model}
We can also analyze our algorithm under the PRAM model. 
Algorithm \ref{algo:parAlgo} is easily translated into the EREW model, 
reaching $\Oh{m/\log m}$ processors and $\Oh{\log m}$ time, dominated by the 
parallel list ranking of line 20, the expansion from $n$ to $m$ processors in 
line 29, and the construction of succinct structures in line 36. The
construction in Section~\ref{sec:extrastructs}, of the structures that speed up 
$\degree$ and $\neighbour$ queries, is also easily carried out in the EREW
model within those bounds, except for the sorting of the edges of $G'$. This
can be done in $\Oh{\log m}$ time with $\Oh{m}$ processors in the EREW model
\cite{Col88}, and in $\Oh{\log^2 m}$ time with $\Oh{m/\log m}$ processors in 
the CREW model \cite{BN89}. The postprocessing we have described in this 
section, once the spanning tree is built, also runs in $\Oh{\log m}$ time and
$\Oh{m/\log m}$ EREW processors.

The most costly part of the process is likely to be the construction of the 
spanning tree. The best PRAM results we know of are
$\Oh{\log^2 m \log^* m}$ time and $\Oh{m}$ processors in the EREW model
\cite{Sha88}, $\Oh{\log^2 m}$ time and $\Oh{m/\log m}$ processors in the 
arbitrary CRCW model \cite{KTT95}, and $\Oh{\log m}$ time and $\Oh{m^3}$ 
processors in the same model \cite{Hag90}.

\begin{theorem}
\label{thm:pram}
The compact representation introduced in Theorem~\ref{thm:main} of a 
conncected planar graph embedding $G$ with $m$ edges
can be constructed under the PRAM EREW model with $\Oh{m}$ processors and
$\Oh{\log^2 m \log^* m}$ time, and under the PRAM arbitrary CRCW model with
$\Oh{m/\log m}$ processors and $\Oh{\log^2 m}$ time, or $\Oh{m^3}$ processors
and $\Oh{\log m}$ time. 
\end{theorem}


\section{Experiments}
\label{sec:experiments}

We implemented the data structure construction and queries in C and compiled 
it using GCC 5.4. For the parallel construction we used Cilk Plus extension, 
an implementation of the DyM model.
We build only the basic structures, excluding those to speed up operations 
{\tt degree} and {\tt neighbour}.
The code and data needed to replicate our results are available
at \url{http://www.dcc.uchile.cl/~jfuentess/pemb/}.

The experiments were carried out on a NUMA machine with two NUMA nodes. Each
NUMA node includes a 14-core Intel\textregistered{} Xeon\textregistered{} CPU
(E5-2695) processor clocked at 2.3GHz. The machine runs Linux
4.4.0-83-generic, in 64-bit mode. The machine has per-core L1 and L2
caches of sizes 64KB and 256KB, respectively and a per-processor shared L3 cache
of 35MB, with a 768GB DDR3 RAM memory (384GB per NUMA node), clocked at
1867MHz. Hyperthreading was enabled, giving a total of 28 logical cores per NUMA
node.

\subsection{Datasets}

Our experiments ran on real and artificial datasets with different numbers of 
nodes. The datasets are shown in Table~\ref{tbl:datasets}. For the artificial
datasets we generated points $(x,y)$ with the function {\tt rnorm} of {\tt R}.%
\footnote{The {\tt rnorm} function generates random numbers with normal 
distribution given a mean and a standard deviation. In our case, the $x$ and 
$y$ components were generated using mean $0$ and standard deviation $10000$. 
For more information about the {\tt rnorm} function, visit
  \url{https://stat.ethz.ch/R-manual/R-devel/library/stats/html/Normal.html}}
The real dataset, {\tt wc}, corresponds to the coordinates of $2,243,467$ 
unique cities in the world.\footnote{The dataset containing the coordinates 
was created by {\em MaxMind}, available from
  \url{https://www.maxmind.com/en/free-world-cities-database}. The original
  dataset contains $3,173,959$ cities, but some of them have the same
  coordinates. We selected the $2,243,467$ cities with unique coordinates to
  build our dataset {\tt wc}.} 
From those real or generated points, we obtained a
{\em Delaunay Triangulation} using {\em Triangle}, a 
software for the generation of meshes and triangulations\footnote{Available 
at \url{ http://www.cs.cmu.edu/~quake/triangle.html}. Our triangulations were 
generated  using the options {\tt -cezCBVPNE}.}. 
Finally, we generated planar embeddings from the Delaunay triangulations with
the {\em Edge Addition Planarity Suite}\footnote{Available at
  \url{https://github.com/graph-algorithms/edge-addition-planarity-suite}. Our
  embeddings were generated using the options {\tt -s -q -p}.}. The minimum and
maximum degree of the dataset {\tt wc} was 3 and 36, respectively. For the rest
of the datasets, the minimum degree was 3 and the maximum degree was 16. 

\begin{table}[t]
\begin{center}
   \begin{tabular}{r@{\hspace{3ex}}l@{\hspace{3ex}}r@{\hspace{3ex}}r@{\hspace{3ex}}}
     \toprule 	
     & Dataset & Vertices ($n$) & Edges ($m$) \\ 
     \midrule 	
     1 & wc & 2,243,467  & 6,730,395 \\ 
     2 & pe5M & 5,000,000  & 14,999,983 \\ 
     3 & pe10M & 10,000,000  & 29,999,979 \\ 
     4 & pe15M & 15,000,000  & 44,999,983 \\ 
     5 & pe20M & 20,000,000  & 59,999,975 \\ 
     6 & pe25M & 25,000,000  & 74,999,979 \\ 
     \bottomrule 	
   \end{tabular} 
\end{center}
\caption{Datasets used in our experiments.}
\label{tbl:datasets}
\end{table}

\subsection{Space usage}

There are no other implemented compact representations of planar embeddings. 
In this subsection we aim to show that representations designed for other
kinds of graphs are indeed not competitive for this graph family. We compare our
compact representation with four solutions designed to compress
Web graphs, social networks and planar graphs
~\cite{a2031031,Boldi:2011:LLP:1963405.1963488,BLN14,BBK03}, and with one
parallel framework for processing general graphs in compressed
form~\cite{7149297}. The three solutions for Web graphs and social networks 
require reordering the vertices of the graph.
The solution of Apostolico and Drovandi~\cite{a2031031}
({\sc AD}) enumerates the vertices through
a BFS traversal of the graph. The reordering induces two useful
properties: {\em locality} (a vertex with index $i$ will have neighbours with
indexes close to $i$), and {\em similarity} (vertices with similar index will
have similar adjacency lists). Thus, the vertices and their adjacency lists are
compressed following the ordering induced by the BFS traversal. The solution of
Boldi {\it et al.}~\cite{Boldi:2011:LLP:1963405.1963488} ({\sc BRSV}) 
reorders the nodes based on a clustering algorithm called {\em Layered
  Label Propagation (LLP)}. The LLP algorithm is used in combination with the 
{\em WebGraph} framework \cite{Boldi:2004:WFI:988672.988752}.
Brisaboa~{\it et al.}~\cite{BLN14} proposed the {\em
  k\textsuperscript{2}-tree} structure for graph compression. The
k\textsuperscript{2}-tree is a compact tree representation of the adjacency
matrix of a graph. The structure exploits the clustering of the edges in the
adjacency matrix, representing large empty areas of the matrix
efficiently. The clustering is dependent on the ordering of the vertices of the
graph. In our comparison, we used the k\textsuperscript{2}-tree structure
combined with the BFS traversal of~\cite{a2031031}, as suggested
by Hern\'andez and Navarro \cite{Hernandez2014}. Blandford~{\it et al.} \cite{BBK03} proposed
a compact representation based on graph separators ({\sc GS}). To construct the
compact representation, the vertices of the graph must be renumbered. The new
numbering is computed recursively, decomposing the graph by the computation of
graph separators. The sequence of computed separators generate the new
numbering. After the renumbering step, adjacent vertices tend to be close in the numbering. The
representation takes advantage of that and reorders the adjacency list of each
vertex, storing the difference between consecutive neighbours. Finally, the
adjacency lists are encoded space-efficiently. In our
experiments, we use the {\em child-flipping} heuristic \cite{BBK03} to compute
the numbering of the vertices and {\em snip} code to encode the adjacency
lists, which was the best among the choices we tested.
Shun~{\it et al.} \cite{7149297}
introduced {\em Ligra}+, a lightweight graph processing framework
for shared-memory multicore machines. In Ligra+, the graph is stored in
compressed form, by compressing the adjacency list of each vertex. The
adjacency list of each vertex is sorted in increasing order and then the
consecutive differences are run-length encoded. Finally, we also consider a
plain representation {(\sc Plain)} composed by an array of length $2m$,
representing the concatenation of the adjacency lists, and an array of length
$n$, representing the beginning of the adjacency list of each vertex.

\begin{table}[t]
\begin{center}
   \begin{tabular}{l@{\hspace{3ex}}ccccccc}
     \toprule 	
     Dataset & {\sc Plain} & {\sc Ligra+} & {\sc BRSV} & {\sc AD} & {\sc
       k\textsuperscript{2}-tree} & {\sc GS} & {\sc Pemb}\\ 
     \midrule 	
     wc & 74.67 & 52.50 & 14.57 & 14.73 & 16.40 & 14.88 & 6.00 \\ 
     pe5M & 74.67 & 52.99 & 14.97 & 14.14 & 15.33 & 15.12 & 5.93 \\ 
     pe10M & 74.67 & 53.15 & 15.03 & 14.33 & 14.73 & 15.12 & 5.93\\ 
     pe15M & 74.67 & 53.20 & 15.04 & 14.38 & 14.38 & 15.12 & 5.72\\ 
     pe20M & 74.67 & 53.24 & 15.07 & 14.43 & 14.15 & 15.12 & 5.93\\ 
     pe25M & 74.67 & 53.32 & 15.11 & 14.50 & 13.96 & 15.14 & 5.80\\ 
     \bottomrule 	
   \end{tabular} 
\end{center}
\caption{Bits per edge (bpe) of the plain representation, alternative 
compressed graph representations, and ours.}
\label{tbl:bpe}
\end{table}

Table~\ref{tbl:bpe} shows the bits per edge ({\em bpe}) of all the
representations, where our solution is called {\sc Pemb}, for {\em planar
  embedding}. In the table, we consider four bytes for each vertex and edge in
the plain representation, equivalent to an integer number in common programming
languages. Our compact
representation reaches the best results, using less than half of the space
of its closest competitor. Note that the other results, using widely different
techniques, obtain very close results, around 15 bpe. This seems to suggest
that exploiting planarity is the key to obtain a drastic reduction in space.
Our results, with at most $6$ bpe, are in accordance
with the $4m+o(m)$ bits of Theorem~\ref{thm:main}.\footnote{We can get closer
to 4 bpe by sparsifying the sublinear-size structures used to query bitvectors and
parentheses, thus trading space for query time.} Notice that due to the
reordering needed by the other representations, they are not suitable for
representing a particular planar embedding.

\subsection{Query times}

\begin{table}[t]
\begin{center}
   \begin{tabular}{l@{\hspace{1.3ex}}c@{\hspace{1.3ex}}c@{\hspace{1.3ex}}c@{\hspace{1.3ex}}c@{\hspace{1.3ex}}c@{\hspace{1.3ex}}c@{\hspace{1.3ex}}c@{\hspace{1.3ex}}c@{\hspace{1.3ex}}}
     \toprule
     \multirow{2}{*}{Dataset} & \multicolumn{4}{c}{Plain} &
     \multicolumn{4}{c}{Compact}\\
     \cmidrule(r){2-5}\cmidrule(r){6-9}
      & {\tt degree} & {\tt listing} & {\tt face} & {\tt dfs} & {\tt degree} & {\tt listing} & {\tt face} & {\tt dfs}\\ 
     \midrule 	
     wc & 0.01 & 0.12 & 0.35 & 0.51$s$ & 4.04 & 20.01 & 8.28 & 14.87$s$\\
     pe5M & 0.02 & 0.14 & 0.51 & 1.39$s$ & 4.24 & 20.65 & 8.55 & 34.61$s$\\
     pe10M & 0.03 & 0.14 & 0.60 & 2.65$s$ & 4.41 & 21.24 & 8.82 & 70.05$s$\\
     pe15M & 0.03 & 0.15 & 0.62 & 4.51$s$ & 4.51 & 21.61 & 8.98 & 106.50$s$\\
     pe20M & 0.03 & 0.15 & 0.64 & 5.66$s$ & 4.60 & 21.77 & 9.17 & 142.20$s$\\
     pe25M & 0.03 & 0.15 & 0.64 & 7.46$s$ & 4.64 & 22.15 & 9.40 & 181.09$s$\\
     \midrule 	
     lim{25M} & 9.31$ms$ & 42.62$ms$ & 12.79$ms$ & - & 2.04$\mu s$ &
     8.47$\mu s$ & 5.26$\mu s$ & 150.20$s$ \\ 
     \bottomrule 	
   \end{tabular} 
\end{center}
\caption{Median times of degree, listing and face queries, and the DFS
  traversal. All the values are in microseconds ($\mu s$), except the {\tt dfs}
  columns and the lim{25M} row, which explicitly indicate
$\mu s$, $ms$ or $s$ (seconds).}
\label{tbl:queries}
\end{table}


We test the time to carry out the three basic queries introduced in
Section~\ref{sec:structure}: {\tt degree}, {\tt listing} and {\tt
  face}. Additionally, 
we test a more complex operation: a depth-first search traveral, {\tt dfs},
starting from an arbitrary vertex and using a stack. We solve
{\tt degree} by sequentially traversing the edges, as we have not built the 
extra data structures to speed up this query. Observe
that, given an adjacency list representation, answering {\tt degree} and {\tt listing} queries is
straightforward. We measured the time of queries {\tt 
degree} and {\tt listing} 10 times per vertex, {\tt face} 10 times per edge, and
{\tt dfs} 10 times for 30 random vertices.
Table~\ref{tbl:queries} shows the median time per query, both for the 
plain representation and for our compact representation. The
plain representation answers {\tt degree}
and {\tt listing} queries $200$ and $150$ times faster than the compact
representation, respectively. This result was expected, since the plain
representation we use already has the list of neighbours in counter-clockwise
order. For the {\tt face} query and the {\tt dfs} traversal, the adjacency list
representation is only $16$ and $26$ times faster, respectively.

This slowdown is the price of a representation that uses about 13 times less 
space, that is, it could hold graphs 13 times larger in main memory. To
illustrate the effect of holding the compressed graph representation in main
memory versus having to handle it on disk, we replicate the experiments in a 
machine with artificially limited memory. For these new experiments we use the pe25M dataset,
whose plain representation requires 668MB, whereas its compact
representation needs only 52MB. 
The machine was set to use at most 600MB of RAM memory\footnote{The computer tested is a
  Intel\textregistered{} Core\textsuperscript{TM} i7-7500U CPU, with four
  physical cores running at 2.70GHz. The computer runs Linux
4.8.0-53-generic, in 64-bit mode. This machine has per-core L1 and L2 caches of sizes 64KB and 256KB,
respectively, and a shared L3 cache of 4MB, with a 8GB DDR4 RAM. To reduce the
size of the available physical memory, we set the {\tt mem} parameter of the
Linux Kernel to {\tt mem=600MB}.}, just slightly less than the necessary to
hold the whole input representation. The results are shown in the last row of
Table~\ref{tbl:queries}. For {\tt degree} query, the compact representation is
around 4,500 
times faster than the plain representation. For the {\tt listing} query, the
difference is around 5,000 times. For the {\tt face} query, the compact
representation is around 2,400 times faster than the plain representation.  We aborted
the experiment on {\tt dfs} for the adjacency list representation after two hours; a
projection of the other results suggests that more than a day would have been
needed.

Thus, the compact representation pays off when it is the key to allow
holding the graph in main memory.


\subsection{Parallel construction}

\begin{table}[t]
\begin{center}
  \small
  \setlength{\tabcolsep}{0pt}
  \begin{tabular}{c@{\hspace{.8em}}r@{\hspace{.8em}}r@{\hspace{.8em}}r@{\hspace{.8em}}r@{\hspace{.8em}}r@{\hspace{.8em}}r}
    \toprule
    $p$ & {\tt wc} & {\tt pe5M} & {\tt pe10M} & {\tt pe15M} & {\tt pe20M}
    & {\tt pe25M}\\
    \midrule
        {\tt seq} & 2.93 & 7.36 & 15.46 & 23.61 & 31.76 & 40.01 \\
        1 & 3.56 & 8.93 & 18.77 & 28.78 & 39.33 & 49.20 \\
        2 & 2.24 & 5.15 & 10.74 & 16.24 & 21.88 & 27.26 \\
        4 & 1.33 & 2.98 & 5.94 & 8.94 & 12.31 & 15.25 \\
        8 & 0.76 & 1.73 & 3.43 & 5.03 & 6.53 & 8.14 \\
        12 & 0.54 & 1.22 & 2.43 & 3.59 & 4.70 & 5.84 \\
        16 & 0.43 & 1.00 & 1.86 & 2.80 & 3.66 & 4.54 \\
        20 & 0.36 & 0.80 & 1.63 & 2.37 & 3.13 & 3.88 \\
        24 & 0.31 & 0.72 & 1.41 & 2.05 & 2.83 & 3.44 \\
        28 & 0.27 & 0.65 & 1.23 & 1.97 & 2.50 & 3.08 \\
        32 & 0.27 & 0.60 & 1.21 & 1.77 & 2.30 & 2.88 \\
        36 & 0.27 & 0.59 & 1.12 & 1.69 & 2.25 & 2.74 \\
        40 & 0.24 & 0.54 & 1.05 & 1.57 & 2.07 & 2.60 \\
        44 & 0.23 & 0.52 & 0.99 & 1.49 & 1.96 & 2.46 \\
        48 & 0.22 & 0.49 & 0.95 & 1.42 & 1.91 & 2.35 \\
        52 & 0.22 & 0.48 & 0.92 & 1.38 & 1.82 & 2.28 \\
        56 & 0.22 & 0.47 & 0.91 & 1.34 & 1.78 & 2.23 \\
        \bottomrule
  \end{tabular}
\end{center}
\caption{Running times of the parallel construction algorithm in seconds.}
\label{tbl:parallelTimes}
\end{table}

We now evaluate the performance of our parallel construction. 
In our implementation of the parallel spanning tree algorithm of Bader and 
Cong \cite{BaderCong2005}, to limit the worst case, we included a treshold of 
$\Oh{m/p}$ elements in the stack size of each thread. Each time a thread has 
more nodes that the threshold, it creates a new parallel thread with half of 
its stack. Additionally, we also return, for each node, the reference to 
its parent. This yields better performance than forcing the first edge of each 
node to lead to its parent. 

Additionally, we implemented a 
sequential algorithm called {\tt seq}, which corresponds to a sequential DFS
algorithm to build the spanning tree, followed by the serialization of the 
parallel algorithm. To serialize a parallel algorithm in the DyM model,
we replaced each {\bf parfor} keyword for the {\bf for} keyword and deleted the
{\bf spawn} and {\bf sync} keywords.
Each data point is the median of 15 measurements.

\begin{figure}[t]
\centerline{\includegraphics[width=0.7\textwidth]{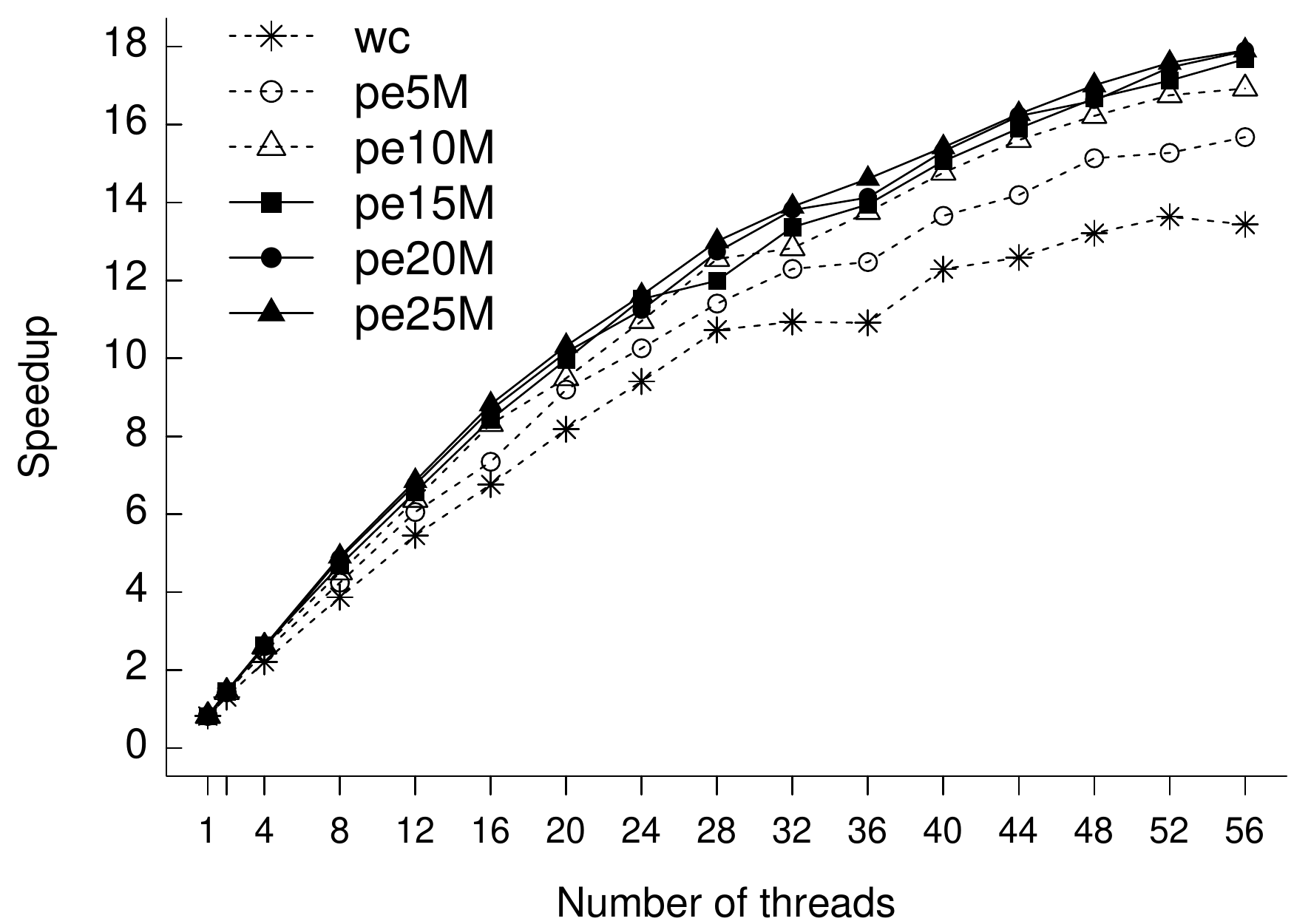}}
\caption{Speedup of the parallel algorithm.}
\label{fig:speedup}
\end{figure}

Table~\ref{tbl:parallelTimes} shows the running times obtained in our
experiments, and Figure~\ref{fig:speedup} shows the speedups compared with the
{\tt seq} algorithm. On average, the {\tt seq} algorithm took about 82\% of the
time obtained by the parallel algorithm running with 1 thread. With $p\geq 2$,
the parallel algorithm shows better times than the {\tt seq} algorithm. We
observe an almost linear speedup up to $p=24$, with an efficiency of at least
40\% for the smaller datasets and almost 50\% for the bigger ones.
With $p=28$ the speedup has a slowdown, due 
to the topology of our machine. Up to 24 cores, all the threads were running in
the same NUMA node. With $p\geq 28$, both NUMA nodes are used, which implies
higher communication costs. The communication costs intra NUMA nodes are lower
than the communication costs inter NUMA nodes \cite{Drepper2007}. In particular,
the case of $p=28$ also uses both NUMA nodes, since at least one core on our
machine was available to OS processes. For $p=56$, the {\tt wc} dataset exhibits
an efficiency of only 24\%, as it is the smallest one. For the bigger
datasets, the lowest efficiency is 32\%. 

The running times and speedups reported in Table~\ref{tbl:parallelTimes} and
Figure~\ref{fig:speedup} include the construction of 
bitvectors and balanced 
parentheses sequences, to support rank, select, parent, and match
operations. To measure the efficiency of our algorithm, without the influence of the
construction of those additional data structures, we repeated all the
construction experiments, excluding the additional data structures.
In the new experiments, we observed that the speedup increases on average 2.7\%
for $p\leq 24$ and 3.2\% for $p\geq 28$, reaching a maximum speedup of 18.8,
compared to the values reported in Figure~\ref{fig:speedup}.

\begin{table}
\begin{center}
  \small
  \setlength{\tabcolsep}{0pt}
  \begin{tabular}{c@{\hspace{.8em}}r@{\hspace{.8em}}r@{\hspace{.8em}}r@{\hspace{.8em}}r@{\hspace{.8em}}r@{\hspace{.8em}}r@{\hspace{.8em}}r}
    \toprule
    $p$ & {\tt 45Me} & {\tt 50Me} & {\tt 55Me} & {\tt 60Me} & {\tt 65Me} & {\tt
      70Me} & {\tt 75Me} \\
    \midrule
        {\tt seq} & 33.20 & 34.63 & 36.16 & 37.17 & 38.05 & 38.99 & 40.01 \\
        1 & 38.35 & 40.74 & 43.06 & 44.87 & 46.38 & 47.85 & 49.20 \\
        2 & 21.42 & 22.86 & 24.30 & 25.17 & 26.35 & 26.85 & 27.26 \\
        4 & 12.10 & 12.96 & 13.50 & 14.31 & 14.50 & 15.16 & 15.25 \\
        8 & 6.45 & 6.89 & 7.23 & 7.51 & 7.76 & 7.88 & 8.14 \\
        12 & 4.62 & 4.91 & 5.12 & 5.46 & 5.64 & 5.69 & 5.84 \\
        16 & 3.59 & 3.85 & 4.05 & 4.20 & 4.28 & 4.45 & 4.54 \\
        20 & 3.05 & 3.28 & 3.41 & 3.59 & 3.66 & 3.80 & 3.88 \\
        24 & 2.71 & 2.86 & 3.01 & 3.08 & 3.23 & 3.30 & 3.44 \\
        28 & 2.45 & 2.62 & 2.68 & 2.77 & 2.90 & 2.94 & 3.08 \\
        32 & 2.30 & 2.49 & 2.57 & 2.69 & 2.75 & 2.81 & 2.88 \\
        36 & 2.19 & 2.35 & 2.38 & 2.50 & 2.58 & 2.64 & 2.74 \\
        40 & 2.04 & 2.16 & 2.25 & 2.33 & 2.43 & 2.48 & 2.60 \\
        44 & 1.94 & 2.04 & 2.13 & 2.21 & 2.28 & 2.34 & 2.46 \\
        48 & 1.83 & 1.93 & 2.04 & 2.12 & 2.18 & 2.24 & 2.35 \\
        52 & 1.77 & 1.86 & 1.95 & 2.05 & 2.10 & 2.16 & 2.28 \\
        56 & 1.71 & 1.82 & 1.90 & 2.00 & 2.06 & 2.14 & 2.23 \\
        \bottomrule
  \end{tabular}
\end{center}
\caption{Running times of the parallel construction algorithm varying the edge
  density for the dataset {\tt pe25M}. The running times are measured in
  seconds.}
\label{tbl:edgeDensity}
\end{table}
Table~\ref{tbl:edgeDensity} shows the running time for different edge densities
of the dataset pe25M, and Figure~\ref{fig:speedup-reduced} shows the
corresponding speedups compared with the algorithm {\tt seq}. The different
densities are generated by deleting $x$ 
million edges from the dataset pe25M, with $x\in\{5,10,15,20,25,30\}$. If
several components are generated, we reconnect them by restoring
one edge between two components and then choosing new edges to be deleted. Thus, we report results for
45 to 75 ({\tt 45Me} to {\tt 75Me}) million edges. The
dataset  {\tt 75Me} corresponds to the original dataset pe25M. We observe a
decrease in the running time for all values of $p$, according to the decrease in
the number of edges. With respect to {\tt 75Me}, the rest of the datasets show a
greater decrease in the running time for increasing values of $p$, reaching
speedups of up to 19.5 for {\tt 45Me}. In the case of datasets with the same number of edges
(see columns {\tt pe15M} and {\tt pe20M} in Table~\ref{tbl:parallelTimes}, and
columns {\tt 45Me} and {\tt 60Me} in Table~\ref{tbl:edgeDensity}), the datasets
with higher number of vertices show higher running times. Comparing
Figures~\ref{fig:speedup} and \ref{fig:speedup-reduced}, we observe that our
algorithm 
scales similarly for triangulated and non-triangulated graphs.

\begin{figure}[t]
\centerline{\includegraphics[width=0.7\textwidth]{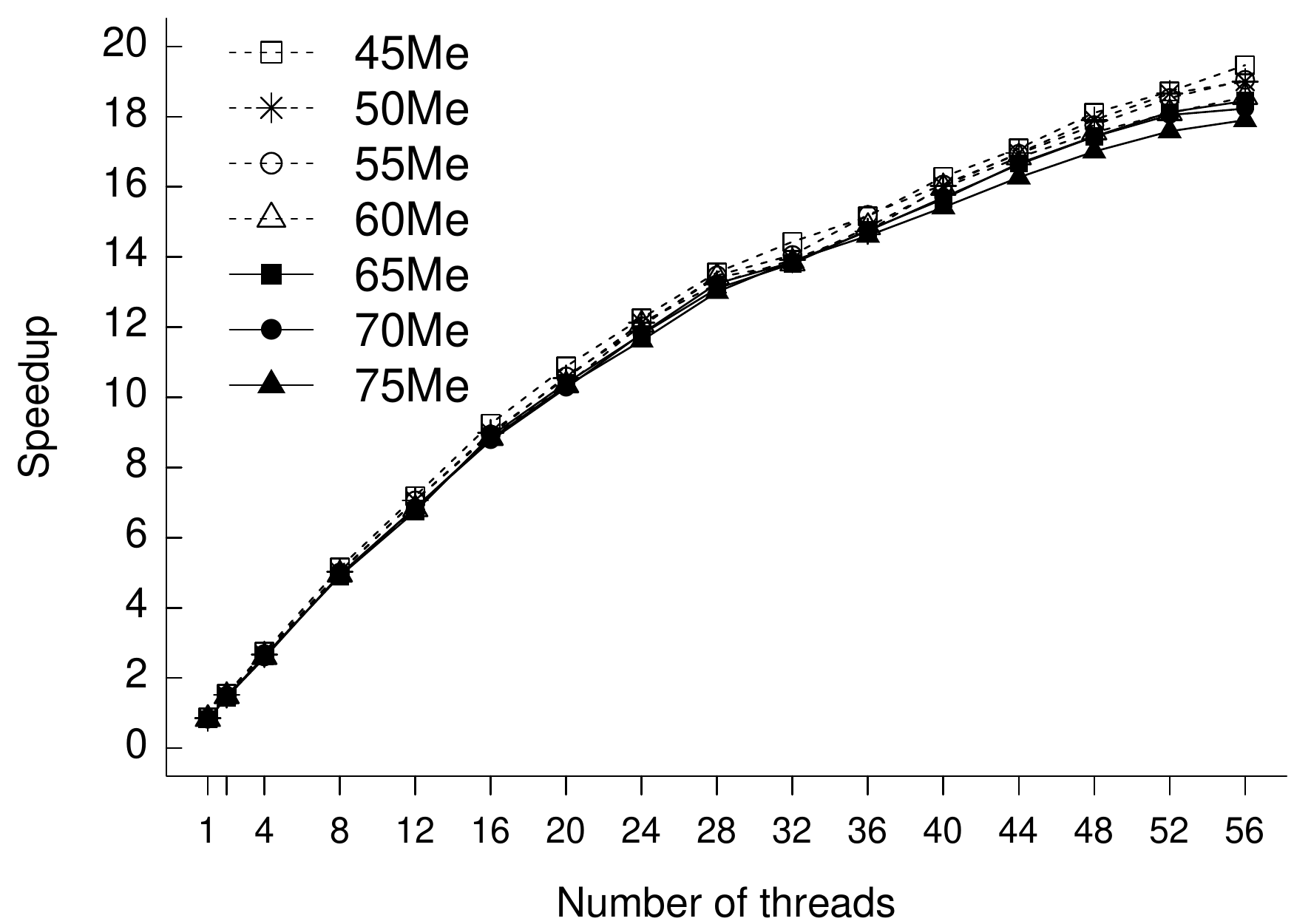}}
\caption{Speedup of the parallel algorithm varying the edge density for the
  dataset {\tt pe25M}.}
\label{fig:speedup-reduced}
\end{figure}

Figure~\ref{fig:memory} shows the memory consumption of our algorithm. 
Specifically, the figure shows for each dataset the space used by its
adjacency list representation ({\tt inputGraph}), the peak consumption of our
construction ({\tt peakMem}) in addition to the input and the output, the space
of its plain representation ({\tt plainGraph}), and the size of its
compact representation ({\tt compGraph}). The plain representation, consisting
of an array of edges of length $2m$ and an array of vertices of length $n$,
is enough to navigate
a graph, but for the construction we need more information about the embedding 
of the input graph. This richer adjacency list representation is what we call 
{\tt inputGraph}. To measure the peak consumption, we use {\tt
  malloc\_count}\footnote{Timo Bingmann. Malloc\_count - Tools for runtime memory usage
  analysis and profiling. URL:\url{https://panthema.net/2013/malloc_count/}.
  Last accessed: August 08, 2017.},
which monitors the memory allocated and released
with {\tt malloc} and {\tt free}, respectively, and reports the peak usage.
The observed peak consumption equals the size of the arrays $LE$, $D_{pos}$ and
$D_{edge}$. Compared with the space consumption of the input 
adjacency list representation, our implementation uses 73\% of extra space. The
final compact representation uses about 8\% of the plain representation, as we have seen.

\begin{figure}[t]
\centering
    \includegraphics[width=0.7\textwidth]{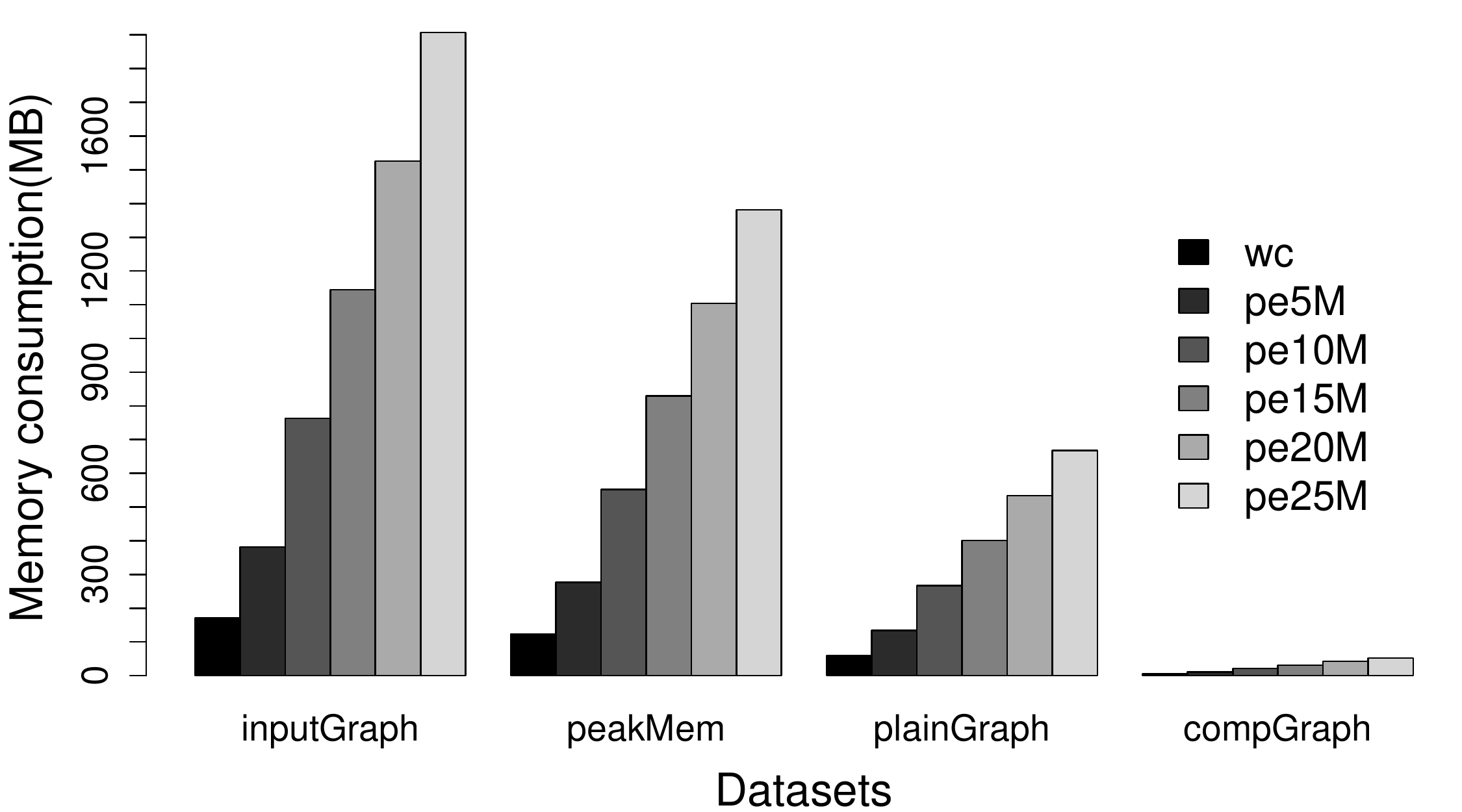}
\caption{Memory consumption of the parallel algorithm and the final compact
structure.}
\label{fig:memory}
\end{figure}

\section{Conclusions and future work}
\label{sec:future}

Tur\'an's representation of planar embeddings \cite{Turan1984} is much simpler
than the known alternatives and encodes any planar embedding of $m$ edges in
just $4m$ bits, close to the lower bound of $3.58m$ bits. In this paper we have
shown how to add $o(m)$ bits to this encoding in order to support fast 
nagivation and queries of the graph, in constant time for the most fundamental
operations. While there are asymptotically optimal representations
\cite{BlellochFarzan2010}, the simplicity of Tur\'an's encoding enabled us to
introduce the first actual implementation of such a compact data structure,
where the basic navigation operations are solved within microseconds. Further,
the structure can be built at a rate of about one microsecond per edge, and
the construction can be parallelized with linear speedup and an efficiency
near 50\%. Our parallel construction algorithm has linear work and logarithmic
span on the dynamic multithreaded model once a spanning tree of the embedding
is computed.

One intriguing question is about the queries we do not support in constant time.
Some previous representations \cite{MR01,ChiangLinLu2005,BlellochFarzan2010}
can compute the degree of a node in $\Oh{1}$ time, whereas we can handle any
superconstant time. Similarly, they can answer neighbour
queries in $\Oh{1}$ time, whereas our structure needs superlogarithmic time.
The representation closest to ours \cite{ChiangLinLu2005} uses the same
technique of two types of parentheses, but the arrangement of the parentheses
follows a so-called orderly spanning tree. While much more complex to build and unable
to represent some embeddings, such spanning tree induces a certain regularity on
the representation of the edges leaving each node, which allows determining in
constant time the number of such edges, and whether two nodes are connected.
It is an interesting question whether we can find a simpler arrangement that
retains those properties.

Another future research line is how to make our data structure dynamic. We can
use a scheme inspired by Munro {\it et al.}~\cite{MNV15}.
Suppose we store our static data structure and
a dynamic buffer that contains information about edges that have been added or
deleted.  If we want to know if an edge is present, we check our static data
structure and then check the buffer to see if its status has changed.  Once the
buffer becomes too large --- e.g., more than \(m / \lg^\epsilon m\) bits --- we
rebuild our static structure.  Even when updates arrive sequentially, there are
some issues to consider, such as how to quickly report the neighbours of a node
that originally had many edges but has had most of them deleted (perhaps by
moving all the information about a node into the buffer when half its incident
edges have been updated) and how to detect if the graph has become non-planar.
There are more issues when the updates can be made in parallel, since then we may need
locks for nodes and finding a practical design becomes challenging.

Finally, we
believe we can generalize our data structure to store efficiently graphs
that are almost planar, using for example generalizations of the technique of
Fischer and
Peters~\cite{FischerPeters2016} to store graphs that are almost trees.
Of course, it is NP-hard to find the maximum planar subgraph of an arbitrary
graph~\cite{Yannakakis1979}, but there have been recent advances in
approximating it and in practice bridges and tunnels, for example, might already
be identified anyway. 

\section*{Acknowledgments}

    The first author received funding from CORFO 13CEE2-21592
    (2013-21592-1-INNOVA PRODUCCION2013-21592-1). The second author received
    funding from Conicyt Fondecyt grant 3170534. The second, third and fifth
    authors received travel funding from EU grant H2020-MSCA-RISE-2015 BIRDS GA
    No.\ 690941, and funding from Basal Funds FB0001, Conicyt, Chile. The third
    author received funding from Academy of Finland grant 268324. The fourth
    author received funding from NSERC of Canada. The fifth author received
    funding from Millennium Nucleus Information and Coordination in Networks,
    ICM/FIC RC130003. Early parts of this work were done while the third author
    was at the University of Helsinki and while the third and fifth authors were
    visiting the University of A Coru\~na.

    Many thanks to J\'er\'emy Barbay, Luca Castelli Aleardi, Guojing Cong, Arash
    Farzan, Cecilia Hern\'andez, Ian Munro, Pat Nicholson, Romeo Rizzi and Julian Shun for
    fruitful discussions. We thank Susana Ladra and Guy Blelloch for sharing
    their k\textsuperscript{2}-tree and graph separators code with us. We also
    thank Telefonica I+D, in particular, Pablo Garc\'ia, for sharing their
    computing equipment with us. The third author is grateful to the late David
    Gregory for his course on graph theory.

\section*{References}

\bibliography{CGTA}

\end{document}